\documentclass{article}

\usepackage{microtype}
\usepackage{graphicx}
\usepackage{subcaption}
\usepackage{booktabs} 
\usepackage{multirow}
\usepackage{tabularx}
\usepackage{tablefootnote}
\usepackage{hyperref}

\usepackage{algorithm}
\usepackage{algorithmic}

\usepackage[preprint]{icml2026}

\usepackage{amsmath}
\usepackage{amssymb}
\usepackage{mathtools}
\usepackage{amsthm}

\usepackage[capitalize,noabbrev]{cleveref}

\theoremstyle{plain}
\newtheorem{theorem}{Theorem}[section]
\newtheorem{proposition}[theorem]{Proposition}

\theoremstyle{definition}

\theoremstyle{remark}
\newtheorem{remark}[theorem]{Remark}

\usepackage[textsize=tiny]{todonotes}

\icmltitlerunning{Submission and Formatting Instructions for ICML 2026}

\begin{document}

\twocolumn[
  \icmltitle{Low-Complexity Algorithm for Stackelberg Prediction Games \\with Global Optimality}



  \icmlsetsymbol{equal}{*}

  \begin{icmlauthorlist}
    \icmlauthor{Tong Wei}{equal,sch,yyy}
    \icmlauthor{Yangjie Xu}{equal,yyy}
    \icmlauthor{Xinlin Wang}{yyy}
    \icmlauthor{Pin-Han Ho}{comp}
    \icmlauthor{Bhavani Shankar M.R.}{yyy}
    \icmlauthor{Radu State}{yyy}
    \icmlauthor{Björn Ottersten}{yyy}
  \end{icmlauthorlist}

  \icmlaffiliation{yyy}{Interdisciplinary Centre for Security, Reliability and Trust (SnT), University of Luxembourg, L-1855 Luxembourg City, Luxembourg.}
  \icmlaffiliation{comp}{Shenzhen Institute for Advanced Study, University of Electronic Science and Technology of China, Shenzhen, Guangdong, China.}
  \icmlaffiliation{sch}{School of Electrical and Electronic Engineering, Nanyang Technological University, Singapore 639798.}

  \icmlcorrespondingauthor{Tong Wei}{tong.wei@ntu.edu.sg}
  \icmlcorrespondingauthor{Yangjie Xu}{yangjie.xu@uni.lu}

  \icmlkeywords{Stackelberg Prediction Games, Machine Learning, ADMM}

  \vskip 0.3in
]



\printAffiliationsAndNotice{}  

\begin{abstract}
Stackelberg prediction games (SPGs) model strategic data manipulation in adversarial learning via a leader--follower interaction between a learner and a self-interested data provider, leading to challenging bilevel optimization problems. Focusing on the least-squares setting (SPG-LS), recent work shows that the bilevel program admits an equivalent spherically constrained least-squares (SCLS) reformulation, which avoids costly conic programming and enables scalable algorithms. In this paper, we develop a simple and efficient alternating direction
method of multiplier (ADMM) based solver for the SCLS problem. By introducing a consensus splitting that separates the quadratic objective from the spherical constraint, we obtain an augmented Lagrangian formulation with closed-form updates: the primal quadratic step reduces to solving a fixed shifted linear system, the constraint step is a projection onto the unit sphere, and the dual step is a lightweight scaled ascent. The resulting method has low per-iteration complexity and allows pre-factorization of the constant system matrix for substantial speedups. Experiments demonstrate that the proposed ADMM approach achieves competitive solution quality with significantly improved computational efficiency compared with existing global solvers for SCLS, particularly in sparse and high-dimensional regimes.
\end{abstract}

\section{Introduction}
Machine learning (ML) systems increasingly operate in adversarial environments where training datasets are not only passively generated, but strategically shaped by self-interested agents \cite{biggio2018wild,perdomo202performative}. In many high-stakes applications, data providers can observe or infer the learner’s decision rule and then adapt their behavior, creating a feedback loop between model design and data generation \cite{liu2019delayed}. Stackelberg prediction games (SPGs) provide a principled framework for strategic interactions between a learner and an adversarial data provider who manipulates training data after observing the learner’s model \cite{bruckner2011stackelberg,zhou2016modeling,bishop2020advances}. Such leader–follower dynamics arise in a wide range of security-critical learning applications, including spam filtering \cite{dalvi2004adversarial}, malware detection \cite{biggio2013evasion}, and strategic reporting \cite{gast2020linear}.  
However, SPGs yield challenging bilevel least-squares optimization problems that are often non-deterministic polynomial-time (NP) hard, even with the convex objective function \cite{jeroslow1985the,wang2021fast}.

Towards this end, initial investigation focuses on a tractable subclass of Stackelberg prediction games with least-squares losses (SPG-LS), where both the learner and the follower use squared loss, and the follower is regularized to limit the magnitude of feature manipulation \cite{bishop2020advances}.  Early methods for obtaining global solutions reduce SPG-LS to single-level convex programs, such as semidefinite programming (SDP) \cite{bishop2020advances} and second-order cone formulations (SOCP) \cite{wang2021fast}, but these approaches become expensive at scale because they require repeated conic solves and, in some instances, large-scale spectral decompositions. As a result, their practical utility is limited in high-dimensional settings.

Recently, a key breakthrough is the spherically constrained least-squares (SCLS) reformulation \cite{wang2022solving}. By applying a nonlinear change of variables, the bilevel optimization problem of SPG-LS can be equivalently expressed as minimizing a least-squares objective over the unit sphere. This reformulation eliminates the reliance on conic optimization and instead supports iterative algorithms whose main computational work is standard linear-algebra operations, leading to markedly improved scalability in high-dimensional regimes. More recently, SCLS has also gained stronger theoretical support. Under the probabilistic data models, error analyses show that the solution recovered by SCLS achieves bounded estimation error, providing formal justification for its reliability beyond empirical evidence \cite{li2024error}. 

Despite these advances, solving SCLS efficiently and robustly in a large-scale setup remains practically important. Existing SCLS solvers in the literature primarily rely on Krylov-subspace schemes and Riemannian trust-region Newton-type methods which are moderately effective but may require careful tuning depending on the regime, sophisticated manifold optimization implementations, nontrivial linear-algebra primitives and lack of performance guarantee \cite{wang2022solving}. 
This motivates us to explore an alternative solver that (i) exploits the special quadratic structure of SCLS; (ii) admits simple closed-form sub-steps and low-complexity algorithms; (iii) is amenable to large-scale implementations with predictable per-iteration cost.

In this paper, we build on the SCLS reformulation of SPG-LS and develop an low-complexity alternating direction method of multipliers (ADMM) algorithm to globally solve the resulting SCLS problem. Specifically, we recast SCLS into a consensus form by introducing an auxiliary variable to isolate the spherical constraint, leading to an augmented-Lagrangian splitting where each update is available in closed form. Based on this, the $\mathbf{r}$-subproblem reduces to solving a shifted linear system with a \emph{fixed} coefficient matrix, the $\mathbf{s}$-subproblem becomes a projection onto the unit sphere, and the dual step is a standard scaled ascent. This yields a lightweight iterative method whose dominant cost is a pre-factorizable linear solve and cheap vector operations thereafter.

Our main contributions are summarized as follows:
\begin{itemize}
  \item \textbf{Closed-form ADMM for SCLS.}
  We develop a low-complexity ADMM algorithm for the consensus SCLS reformulation with fully explicit closed-form updates, yielding a simple and scalable implementation.
  \item \textbf{Global optimality and convergence based on strong duality.}
  We prove strong duality for the consensus reformulation and show that a Lagrangian saddle point corresponds to an optimal solution of SCLS, making the KKT conditions necessary and sufficient. Under mild regularity conditions, we further show that ADMM monotonically decreases the augmented Lagrangian and converges to a primal-dual limit point that satisfies the global KKT conditions and is therefore globally optimal.
\item \textbf{Low-complexity inverse-free implementation.}
  The computational bottleneck of ADMM is the $\mathbf{r}$-update, which requires solving a linear system with a fixed positive-definite coefficient matrix. Toward this end, we avoid explicit matrix inversion by computing a one-time Cholesky factorization of the constant system matrix and then performing only two triangular solves per iteration. This inverse-free implementation substantially reduces complexity and runtime, especially in large-scale and sparse regimes.
  \item \textbf{Extensive performance validation.}
  Experiments on both synthetic and real datasets show that our method matches the accuracy of state-of-the-art SCLS solvers while substantially reducing runtime, with gains most pronounced in high-dimensional and sparse data regimes.
\end{itemize}

\section{Preliminaries}
\label{sec:preliminaries}

This section introduces the problem formulation of SPG-LS and briefly reviews standard single-level reformulations that enable local or global solutions. 

\begin{figure}[t]
\begin{center}
\centerline{\includegraphics[width=1\columnwidth]{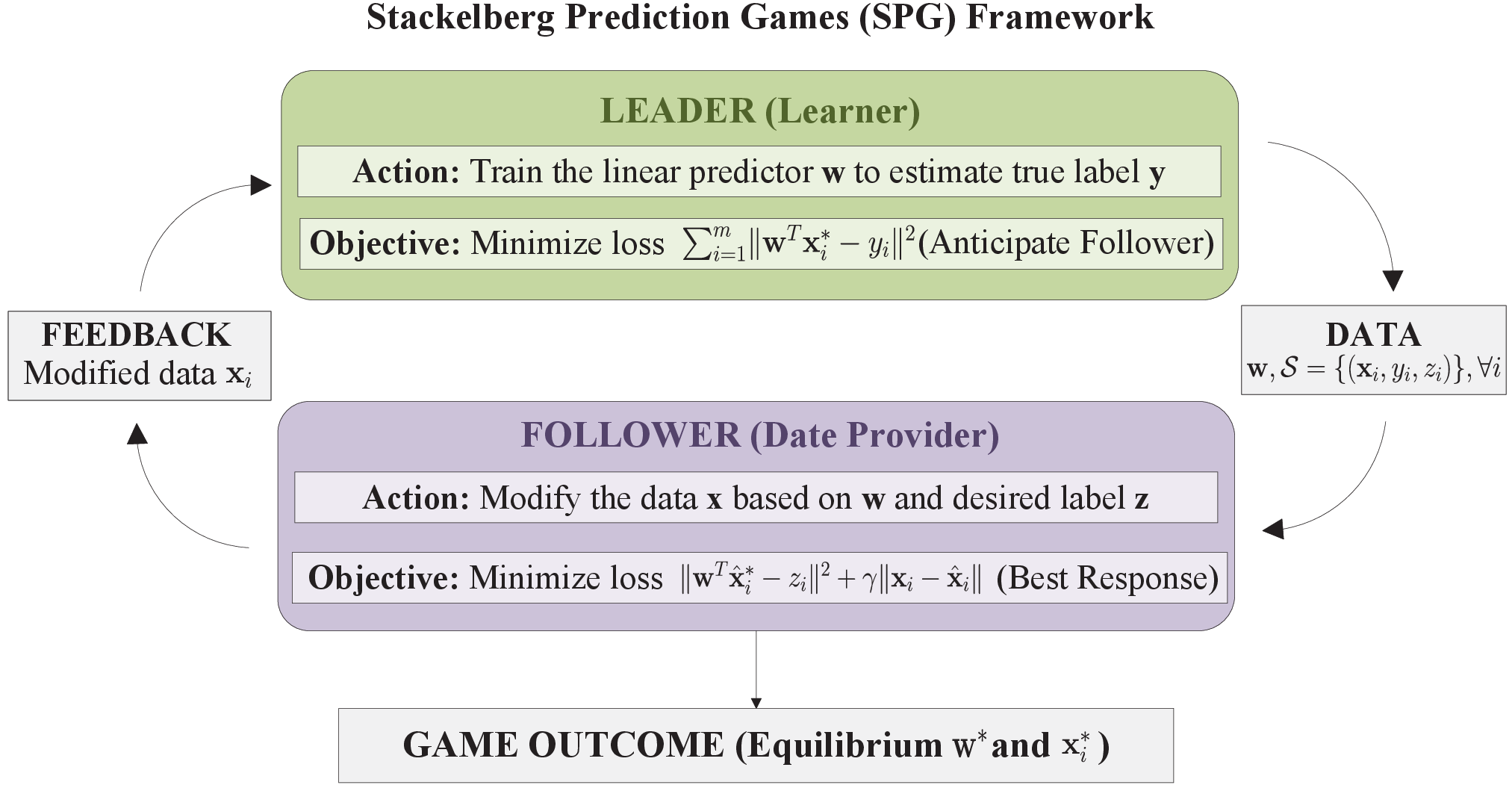}}
    \caption{Illustration of Leader–Follower Interaction in Stackelberg Prediction Games Framework.}
    \label{dig}
\end{center}
\end{figure}

\subsection{SPG-LS Problem setup}

We consider a Stackelberg prediction game between a \emph{learner} (leader) and a \emph{data provider} (follower). The learner first observes an initial training set $\mathcal{S}=\{(\mathbf{x}_i,y_i,z_i)\}_{i=1}^m$,
where $\mathbf{x}_i\in\mathbb{R}^n$ is the feature vector, $y_i\in\mathbb{R}$ is the true label, and $z_i\in\mathbb{R}$ denotes the provider's desired label. Based on $\mathcal{S}$, the learner commits to a linear predictor $\mathbf{w}\in\mathbb{R}^n$ by minimizing its prediction loss while anticipating the provider's response. After observing $\mathbf{w}$, the provider strategically \emph{modifies the features} to better align the learner's predictions with $z_i$, subject to a regularization term that limits the extent of manipulation. The learner is then trained on the resulting modified data, and the interaction continues until reaching an equilibrium $(\mathbf{w}^\star,{\mathbf{x}}^\star)$; See the detailed leader-follower interaction in \textit{Figure~\ref{dig}}.

The setting can be cast as a Stackelberg prediction game \cite{bruckner2011stackelberg,bishop2020advances}, in which a learner and a strategic data provider interact sequentially. The learner first commits to a prediction rule, anticipating how the provider will react. Given the learner's model $\mathbf{w}$, the data provider strategically modifies each feature vector to $\hat{\mathbf{x}}_i$ by trading off (i) forcing the prediction $\mathbf{w}^{T} \hat{\mathbf{x}}_i$ toward $z_i$ and (ii) keeping $\hat{\mathbf{x}}_i$ close to $\mathbf{x}_i$, which can be formulated as the following optimization problem 
\begin{equation}
\mathbf{x}_i^\star(\mathbf{w})\in\arg\min_{\hat{\mathbf{x}}_i\in\mathbb{R}^n}
\bigl(\mathbf{w}^{T} \hat{\mathbf{x}}_i - z_i\bigr)^2+\gamma\|\hat{\mathbf{x}}_i-\mathbf{x}_i\|_2^2,~ i\in[m].
\end{equation}
where $\gamma>0$ denotes the regularization weight. Meanwhile, given the manipulated features $\{\tilde{\mathbf{x}}_i\}_{i=1}^m$, the learner fits a linear predictor by minimizing the squared prediction error as follows: 
\begin{align}
\mathbf{w}^\star = \arg\min_{\mathbf{w}\in\mathbb{R}^n}
\sum_{i=1}^{m}\left(\mathbf{w}^{T} \tilde{\mathbf{x}}_i - y_i\right)^2 .
\end{align}
Let $\mathbf{X}\in\mathbb{R}^{m\times n}$ stack rows $\mathbf{x}_i^{T}$, and let $\mathbf{y},\mathbf{z}\in\mathbb{R}^m$ stack $\{y_i\},\{z_i\}$ and denote by $\mathbf{X}^\star(\mathbf{w})\in\mathbb{R}^{m\times n}$ the matrix stacking $(\mathbf{x}_i^\star(\mathbf{w}))^{T}$. The Stackelberg equilibrium of the two players can be formulated as the following bilevel optimization problem:
\begin{subequations} \label{eq:spg_ls_bilevel}
\begin{align}
&\min_{\mathbf{w}\in\mathbb{R}^n}\ ~\|\mathbf{X}^{\star}\mathbf{w}-\mathbf{y}\|_2^2 \\
&\text{s.t.} ~
\mathbf{X}^\star\in\arg\min_{\hat{\mathbf{X}}\in\mathbb{R}^{m\times n}}
\|\hat{\mathbf{X}}\,\mathbf{w} - \mathbf{z}\|_2^2+\gamma\|\hat{\mathbf{X}}-\mathbf{X}\|_F^2.
\end{align}
\end{subequations}
Note that problem~\eqref{eq:spg_ls_bilevel} is a nonconvex bilevel program and is NP-hard in general~\cite{wang2019global}, even with quadratic losses and simple regularization. The difficulty arises from the coupling between the learner’s decision and the provider’s best response, which induces a nonlinear upper-level objective. We present reformulations that mitigate this issue and enable efficient algorithms for~\eqref{eq:spg_ls_bilevel} in the sequel.

\subsection{Quadratic Fractional Reformulation}
By exploiting a closed-form characterization of the follower’s best response via a Sherman--Morrison-type argument \cite{sherman1950adjustment}, the bilevel problem~\eqref{eq:spg_ls_bilevel} can be equivalently reduced to an unconstrained quadratic fractional program
\begin{equation} \label{eq:qfp_w}
\inf_{\mathbf{w}\in\mathbb{R}^n}\ 
\left\|
\frac{\tfrac{1}{\gamma}(\mathbf{w}^{T} \mathbf{w)\,\mathbf{z} + \mathbf{X}\mathbf{w}}}{1+\tfrac{1}{\gamma}(\mathbf{w}^{T} \mathbf{w)} }-\mathbf{y}
\right\|_2^2.
\end{equation}
Introduce the scalar auxiliary variable $\alpha \triangleq \mathbf{w}^T \mathbf{w}/{\gamma} \ge 0$, an equivalent constrained formulation of \eqref{eq:spg_ls_bilevel} is given
\begin{subequations}\label{eq:qfp_wa}
\begin{align}
\min_{\mathbf{w}\in\mathbb{R}^n,\ \alpha\ge 0} \quad&
v(\mathbf{w},\alpha)\triangleq
\left\|
\frac{\alpha \mathbf{z} + \mathbf{X}\mathbf{w}}{1+\alpha}-\mathbf{y}
\right\|_2^2 \\
\quad\text{s.t.}\quad &
\mathbf{w}^{T} \mathbf{w}=\gamma\alpha.
\end{align}
\end{subequations}

\subsection{SDP Reformulation}
Problem~\eqref{eq:qfp_wa} further admits an exact semidefinite programming (SDP) reformulation; see Theorem~3.3 in~\cite{wang2021fast}. In particular, it can be equivalently written as
\begin{equation}
\sup_{\mu,\lambda\in\mathbb{R}}\ \mu
\quad\text{s.t.}\quad
\mathbf{A}-\mu \mathbf{B}+\lambda \mathbf{C}\succeq \mathbf{0},
\label{eq:sdp_form}
\end{equation}
where $\mathbf{A},\mathbf{B},\mathbf{C}$ are symmetric matrices defined by
\begin{align*}
\mathbf{A} &=
\begin{pmatrix}
\mathbf{X}^T\mathbf{X} & \mathbf{X}^T(\mathbf{z}-\mathbf{y}) & -\mathbf{X}^T\mathbf{y}\\
(\mathbf{z}-\mathbf{y})^T\mathbf{X} & \|\mathbf{z}-\mathbf{y}\|^2 & -(\mathbf{z}-\mathbf{y})^T\mathbf{y}\\
-\mathbf{y}^T\mathbf{X} & -\mathbf{y}^T(\mathbf{z}-\mathbf{y}) & \mathbf{y}^T\mathbf{y}
\end{pmatrix}, \\
\mathbf{B} &=
\begin{pmatrix}
\mathbf{O}_n & \mathbf{0} & \mathbf{0}\\
\mathbf{0}^T & 1 & 1\\
\mathbf{0}^T & 1 & 1
\end{pmatrix}, 
\mathbf{C} =
\begin{pmatrix}
\frac{1}{\gamma}\mathbf{I}_n & \mathbf{0} & -\frac{1}{2}\mathbf{1}_n\\
\mathbf{0}^T & 0 & -\frac{1}{2}\\
-\frac{1}{2}\mathbf{1}_n^T & -\frac{1}{2} & 0
\end{pmatrix}. 
\end{align*}
Herein, $\mathbf{O}_n$ denotes the $n\times n$ all-zero matrix, $\mathbf{I}_n$ is the $n\times n$ identity matrix, and $\mathbf{1}_n$ is the all-ones vector in $\mathbb{R}^n$.

\subsection{SOCP reformulation}
By an invertible congruence transformation, the linear matrix inequality (LMI) in~\eqref{eq:sdp_form} can be converted into an equivalent SOCP~\cite{wang2021fast}. Specifically, there exists an invertible $\mathbf{V}$ such that the transformed matrices admit an \emph{arrow} form
\begin{equation*}
\widetilde{\mathbf{A}}=
\begin{pmatrix}
\mathbf{D} \!&\! \mathbf{b}\\
\mathbf{b}^{T} \!&\! c
\end{pmatrix},
\widetilde{\mathbf{B}}=
\begin{pmatrix}
\mathbf{O}_{n+1} \!&\! \mathbf{0}\\
\mathbf{0}^T \!&\! 4
\end{pmatrix},
\widetilde{\mathbf{C}}=
\begin{pmatrix}
\frac{1}{\gamma}\mathbf{I}_{n+1} \!&\! \mathbf{0}\\
\mathbf{0}^T \!&\! -1
\end{pmatrix},
\end{equation*}
where $\mathbf{D}=\mathrm{Diag}(d_1,\ldots,d_{n+1})$, $\mathbf{b}\in\mathbb{R}^{n+1}$, and $c\in\mathbb{R}$. A generalized Schur-complement argument then yields the SOCP~\cite{wang2021fast}
\begin{subequations} \label{socp_reform}
\begin{align}
\sup_{\mu,\lambda,\{s_i\}}\ \ & \mu \label{socp_reform_obj}\\
\text{s.t. }\ \ &
d_i+\frac{\lambda}{\gamma}\ge 0,\ i=1,\ldots,n+1, \label{socp_reform_c1}\\
& c-4\mu-\lambda-\sum_{i=1}^{n+1}s_i\ge 0, \label{socp_reform_c2}\\
& s_i\!\left(d_i+\frac{\lambda}{\gamma}\right)\ge b_i^2,\ \ s_i\ge 0,\ i\in [n+1]. \label{socp_reform_c3}
\end{align}
\end{subequations}
While \eqref{socp_reform_c1}--\eqref{socp_reform_c3} is polynomial-time solvable by interior-point methods, forming the SOCP typically requires a spectral decomposition of a matrix of order $(n+1)\times(n+1)$, which can be computationally intensive for high-dimensional feature spaces.

\subsection{SCLS Reformulation}\label{subsec:scls}

We start from the quadratic fractional program reformulation of SPG-LS in \eqref{eq:qfp_wa}.
Following~\cite{wang2022solving}, a nonlinear change of variables transforms~\eqref{eq:qfp_wa} into a spherically constrained least-squares (SCLS) problem:
\begin{subequations}
\label{eq:scls}
\begin{align}
\min_{\tilde{\mathbf{w}}\in\mathbb{R}^n,\ \tilde{\alpha}\in\mathbb{R}} \ \ 
& \tilde v(\tilde{\mathbf{w}},\tilde{\alpha})\triangleq
\left\|
\frac{\tilde{\alpha}}{2}\mathbf{z}
\!+\!\frac{\sqrt{\gamma}}{2}\mathbf{X}\tilde{\mathbf{w}}
\!-\!\Bigl(\mathbf{y}\!-\!\frac{\mathbf{z}}{2}\Bigr)
\right\|_2^2 \label{eq:scls_obj}\\
\text{s.t.}\ \ 
& \tilde{\mathbf{w}}^{T}\tilde{\mathbf{w}}+\tilde{\alpha}^2=1. \label{eq:scls_sphere}
\end{align}
\end{subequations}

Note that the two formulations~\eqref{eq:qfp_wa} and~\eqref{eq:scls} are equivalent under an explicit change of variables \cite{wang2022solving}.
Given any feasible $(\mathbf{w},\alpha)$ of~\eqref{eq:qfp_wa}, define
\begin{equation}
\label{eq:map_forward}
\tilde{\mathbf{w}} \ :=\ \frac{2}{\sqrt{\gamma}(1+\alpha)}\,\mathbf{w},
\quad
\tilde{\alpha}\ :=\ \frac{\alpha-1}{\alpha+1}.
\end{equation}
Then $(\tilde{\mathbf{w}},\tilde{\alpha})$ is feasible for~\eqref{eq:scls} and satisfies
$\tilde v(\tilde{\mathbf{w}},\tilde{\alpha})=v(\mathbf{w},\alpha)$.
Conversely, for any feasible $(\tilde{\mathbf{w}},\tilde{\alpha})$ of~\eqref{eq:scls} with $\tilde{\alpha}\neq 1$, define
\begin{equation}
\label{eq:map_backward}
\mathbf{w} \! :=\! \frac{\sqrt{\gamma}}{1-\tilde{\alpha}}\,\tilde{\mathbf{w}},
\quad
\alpha \! :=\! \frac{1+\tilde{\alpha}}{1-\tilde{\alpha}}.
\end{equation}
Then $(\mathbf{w},\alpha)$ is feasible for~\eqref{eq:qfp_wa} and satisfies
$v(\mathbf{w},\alpha)=\tilde v(\tilde{\mathbf{w}},\tilde{\alpha})$.
To obtain a compact representation, we can define
\begin{equation*}
\hat{\mathbf{L}} \ :=\ \begin{pmatrix}\frac{\sqrt{\gamma}}{2}\mathbf{X} & \frac{1}{2}\mathbf{z}\end{pmatrix}\in\mathbb{R}^{m\times(n+1)},
~
\mathbf{r} \ :=\ \begin{pmatrix}\tilde{\mathbf{w}}\\ \tilde{\alpha}\end{pmatrix}\in\mathbb{R}^{n+1}.
\end{equation*}
With these definitions,~\eqref{eq:scls} can be written equivalently as
\begin{subequations}\label{eq:scls_compact}
\begin{align}
\min_{\mathbf{r}\in\mathbb{R}^{n+1}}\ \ 
& f(\mathbf{r})\ \triangleq\ \|\hat{\mathbf{L}}\,\mathbf{r}-(\mathbf{y}-\mathbf{z}/2)\|_2^2 \label{eq:scls_compact_obj}\\
\text{s.t.}\ \ 
& \mathbf{r}^{T}\mathbf{r}=1. \label{eq:scls_compact_sphere}
\end{align}
\end{subequations}
Moreover, the objective admits the quadratic expansion
\begin{align} \label{eq:quad_form}
f(\mathbf{r}) &= \mathbf{r}^{T}\mathbf{H}\mathbf{r} + 2\,\mathbf{g}^{T}\mathbf{r} + p,
\end{align}
where $\mathbf{H} = \hat{\mathbf{L}}^{T}\hat{\mathbf{L}}$, $\mathbf{g}= \hat{\mathbf{L}}^{T}(\mathbf{z}/2-\mathbf{y})$, and $p = (\mathbf{z}/2-\mathbf{y})^{T}(\mathbf{z}/2-\mathbf{y})$.

The compact formulation~\eqref{eq:scls_compact} reveals that SPG-LS is equivalent to minimizing a least-squares objective over the unit sphere, which serves as the starting point for the ADMM algorithm developed in the next section.

\section{Global ADMM Solution}
\label{sec:algorithm}
We can first reformulate problem \eqref{eq:scls_compact} equivalently as 
\begin{equation} \label{problem_ADMM}
 \begin{aligned} 
   p^{\star}:~ \mathop{{\mathrm{minimize}}}\limits_{{\bf r}, {\bf s}} & {\quad} f({\bf r}) \\
    \textrm{subject to} & {\quad} {\bf s}-{\bf r} = {\bf 0}, ~{\bf s}^T{\bf s} = 1,
\end{aligned}  
\end{equation}
where ${\bf s}$ is the additional variable.
We introduce the following theorem to reveal the duality property of \eqref{problem_ADMM}. 

\begin{theorem} The function $f(\mathbf{r})$ is closed, smooth, bounded, and convex. Then, the strong duality for the non-convex problem \eqref{problem_ADMM} holds. 
\label{lem:duality2}
\end{theorem}
\begin{proof}
See Appendix \ref{proof:psd}.
\end{proof}

\begin{remark}
For non-convex optimization problems, the KKT conditions are generally only necessary for global optimality. An important exception arises for quadratic objectives with spherical constraints (e.g., trust-region-type problems), where the associated dual problem is convex and strong duality holds (i.e., the duality gap is zero)~\cite{Boyd_Vandenberghe_2004}. In this case, the Lagrangian admits a saddle point, and any primal-dual pair satisfying the KKT conditions corresponds to a globally optimal solution of~\eqref{problem_ADMM}.
\end{remark}

\subsection{Standard ADMM}
To improve the convergence properties, ADMM generally updates the variables based on the augmented Lagrangian function \cite{wei2026quadratic}.
The augmented Lagrangian of problem \eqref{problem_ADMM} is 
 \begin{flalign}  \label{aug._LG}
  \mathcal{L}_1({\bf r}, {\bf s}, {\bf u}) \!=\! f({\bf r})+\!{\bf u}^T({\bf s}\!-\!{\bf r})\!+\!\frac{\rho}{2}\|{\bf s}\!-\!{\bf r}\|_2^2,
\end{flalign} 
where ${\bf s}^T{\bf s}=1$ denotes the constraint for variable ${\bf s}$, ${\bf u}$ denotes the dual variable, and $\rho>0$ is the penalty parameter. 
Define ${\bf v}=\frac{1}{\rho}{{\bf u}}$, we have  ${\bf u}^T({\bf s}-{\bf r})+\frac{\rho}{2}\|{\bf s}-{\bf r}\|_2^2 
  = \frac{\rho}{2}\|{\bf s}-{\bf r}+{\bf v}\|_2^2-\frac{\rho}{2}\|{\bf v}\|_2^2$.  
Then, we can rewrite the scaled augmented Lagrangian \eqref{aug._LG} as 
 \begin{align}  \label{aug._LG_R}
  \mathcal{L}_2({\bf r}, {\bf s}, {\bf v}) \!= \!f({\bf r})\!+\!\frac{\rho}{2}\|{\bf s}\!-\!{\bf r}\!+\!{\bf v}\|_2^2\!-\!\frac{\rho}{2}\|{\bf v}\|_2^2.
\end{align}
Within the ADMM framework, each iteration involves minimizing the augmented Lagrangian $\mathcal{L}_2({\bf r}, {\bf s}, {\bf v})$ with respect to the primal variables, followed by an update of the dual variable using the dual ascent method \cite{boyd2011distributed}. Specifically, ADMM updates the primal and dual variables sequentially as follows:
\begin{subequations} \label{ADMM_update}
 \begin{align} 
  {\bf r}^{k+1} =&  \mathop{\mathrm{argmin}}\limits_{\bf r}f({\bf r})\!+\frac{\rho}{2}\|{\bf s}^{k}-{\bf r}+{\bf v}^{k}\|_2^2; \label{ADMM_update1}\\
  {\bf s}^{k+1}  =& \mathop{\mathrm{argmin}}\limits_{{\bf s}}\frac{\rho}{2}\|{\bf s}-{\bf r}^{k}+{\bf v}^{k}\|_2^2, ~{\rm{s.t.}}~{\bf s}^T{\bf s}=1;\label{ADMM_update2}\\
  {\bf v}^{k+1} = & \mathop{\mathrm{argmax}}\limits_{\bf v} \frac{\rho}{2}\|{\bf s}^{k}-{\bf r}^{k}+{\bf v}\|_2^2-\frac{\rho}{2}\|{\bf v}\|_2^2. \label{ADMM_update3}
 \end{align}   
\end{subequations}
Subsequently, the solutions for the subproblems in \eqref{ADMM_update} are, respectively, given by 
\begin{subequations} \label{ADMM_solution}
 \begin{align} 
  {\bf r}^{k+1} &= 
\left(\mathbf{H}+
\frac{\rho}{2}\mathbf{I}
\right)^{-1}
\left(-\mathbf{g}+
\frac{\rho}{2}\left(\mathbf{s}^{k}+\mathbf{v}^{k}\right)
\right);
\label{ADMM_solution1}\\
  {\bf s}^{k+1} &= \frac{1}{\|{\bf r}^{k+1}-{\bf v}^{k}\|_2}({\bf r}^{k+1}-{\bf v}^{k}); \label{ADMM_solution2}\\
  {\bf v}^{k+1} &=  {\bf v}^{k}+{\bf s}^{k+1}-{\bf r}^{k+1}. \label{ADMM_solution3}
 \end{align}   
\end{subequations}
Note that a set of closed-form solutions is achieved for the subproblems in \eqref{ADMM_solution1} and \eqref{ADMM_solution2} which highlights the efficiency of the proposed method. Meanwhile, the major complexity of ADMM arises from the matrix inversion in \eqref{ADMM_solution1}. However, it remains unchanged during the iteration procedure. Hence, we can pre-compute it before ADMM iteration.   
\vspace{-0.2cm}
\begin{proposition}\label{prop2}
Let $\{(\mathbf{r}^k,\mathbf{s}^k,\mathbf{v}^k)\}_{k\ge 0}$ be generated by ADMM in~\eqref{ADMM_solution} for problem~\eqref{problem_ADMM} with $\rho>0$. Then:
(i) the primal and dual residuals converge to zero;
(ii) the objective values $\{f(\mathbf{r}^k)\}$ converge;
(iii) the dual iterates $\{\mathbf{v}^k\}$ converge (and $\|\mathbf{v}^{k+1}-\mathbf{v}^k\|\to 0$);
(iv) every limit point is a KKT point of~\eqref{problem_ADMM} (hence a stationary point).
\end{proposition}
\begin{proof}
See Appendix~\ref{proof:pgc}.
\end{proof}

\vspace{-0.2cm}
\begin{remark}
Theorem~\ref{lem:duality2} establishes strong duality for~\eqref{problem_ADMM}, while Proposition~\ref{prop2} shows that the proposed ADMM iterates converge to a primal--dual stationary point. Together, these results imply that any limit point satisfies the KKT conditions of~\eqref{problem_ADMM} and is therefore globally optimal. This aligns with existing global convergence guarantees for quadratic programs with spherical constraints, provided that the penalty parameter $\rho$ is chosen sufficiently large to ensure a sufficient descent property of the augmented Lagrangian; see~\cite{wang2019global}. In practice, however, overly large $\rho$ can slow down ADMM. Developing principled and efficient rules for selecting $\rho$ to further speed up ADMM in this setting remains an interesting direction for future work.
\end{remark}

\begin{algorithm}[t]
  \caption{Global ADMM Algorithm to Solve \eqref{problem_ADMM}}
  \label{alg:ADMM}
  \begin{algorithmic}[1]
    \STATE \textbf{Input:} Penalty parameter $\rho>0$, tolerance $\epsilon>0$, maximum iterations $K_{\max}$,
    initial $\mathbf{s}^{(0)}$ with $\|\mathbf{s}^{(0)}\|_2=1$, and initial scaled dual variable $\mathbf{v}^{(0)}$.
    \STATE \textbf{Output:} $\mathbf{r}^{\star}=\mathbf{r}^{(k)}$, $\mathbf{s}^{\star}=\mathbf{s}^{(k)}$, and $\mathbf{v}^{\star}=\mathbf{v}^{(k)}$.
    \STATE Pre-compute $\mathbf{E}\triangleq (\mathbf{H}+\frac{\rho}{2}\mathbf{I})^{-1}$.
    \STATE Set $k=0$;
    \REPEAT
      \STATE \textbf{Update $\mathbf{r}^{(k+1)}$} by \eqref{ADMM_solution1}:
      \STATE \hspace{1.6em} $\mathbf{r}^{(k+1)} \leftarrow \mathbf{E}\!\left(-\mathbf{g}+\frac{\rho}{2}\left(\mathbf{s}^{(k)}+\mathbf{v}^{(k)}\right)\right)$;
      \STATE \textbf{Update $\mathbf{s}^{(k+1)}$} by \eqref{ADMM_solution2} (projection onto the unit sphere):
      \STATE \hspace{1.6em} $\mathbf{s}^{(k+1)} \leftarrow \dfrac{\mathbf{r}^{(k+1)}-\mathbf{v}^{(k)}}{\left\|\mathbf{r}^{(k+1)}-\mathbf{v}^{(k)}\right\|_2}$;
      \STATE \textbf{Update $\mathbf{v}^{(k+1)}$} by \eqref{ADMM_solution3}:
      \STATE \hspace{1.6em} $\mathbf{v}^{(k+1)} \leftarrow \mathbf{v}^{(k)}+\mathbf{s}^{(k+1)}-\mathbf{r}^{(k+1)}$;
      \STATE Compute primal residual:
      \STATE \hspace{1.6em} $r_{\mathrm{pri}}^{(k+1)} \leftarrow \left\|\mathbf{s}^{(k+1)}-\mathbf{r}^{(k+1)}\right\|_2$;
      \STATE Compute dual residual:
      \STATE \hspace{1.6em} $r_{\mathrm{dual}}^{(k+1)} \leftarrow \rho\left\|\mathbf{s}^{(k+1)}-\mathbf{s}^{(k)}\right\|_2$;
      \STATE $k \leftarrow k+1$;
    \UNTIL{$\max\{r_{\mathrm{pri}}^{(k)},\,r_{\mathrm{dual}}^{(k)}\}\le \epsilon$ or $k\ge K_{\max}$}
    \STATE \textbf{return} $\mathbf{r}^{(k)}$, $\mathbf{s}^{(k)}$, and $\mathbf{v}^{(k)}$;
  \end{algorithmic}
\end{algorithm}

We summarize the algorithm for solving SCLS in \eqref{problem_ADMM} via the standard ADMM in Algorithm~\ref{alg:ADMM}.

\subsection{Low-Complexity ADMM via Cholesky Decomposition}
To further reduce the computational cost of the $\mathbf{r}$-update, we avoid explicitly forming the inverse of the semidefinite matrix in \eqref{ADMM_solution1} and instead solve a linear system via Cholesky factorization.

To further reduce computational complexity, we can use the Cholesky decomposition to decompose the positive definite matrix. Note that 
Specifically, we can implement Cholesky decomposition  
\begin{align} \label{LU_decom}
   {\mathbf F} = \mathbf{H}+\frac{\rho}{2}\mathbf{I} = \mathbf{U}^T\mathbf{U}
\end{align}
where $\mathbf{U}$ is the upper-triangular matrix. Based on the update of ${\bf r}$ in \eqref{ADMM_solution1}, we have 
\begin{align} \label{LU_ADMM_solution}
\mathbf{U}^T\mathbf{U}{\bf r}^{k+1} &= \left(-\mathbf{g}+\frac{\rho}{2}\left(\mathbf{s}^{k}+\mathbf{v}^{k}\right)\right).
\end{align}
Then, we can define an intermediary variable $\mathbf{y}^{k+1}=\mathbf{U}{\bf r}^{k+1}$, and then we have
\begin{align} \label{LU_ADMM_update}
\mathbf{U}^T\mathbf{y}^{k+1} &= \left(-\mathbf{g}+\frac{\rho}{2}\left(\mathbf{s}^{k}+\mathbf{v}^{k}\right)\right),
\end{align}
and 
\begin{align} \label{LU_ADMM_update2}
\mathbf{U}{\bf r}^{k+1}=\mathbf{y}^{k+1}.
\end{align}
To obtain the update of ${\bf r}^{k+1}$, we can solve the triangular system of linear equations in \eqref{LU_ADMM_update} and then \eqref{LU_ADMM_update2} with the much lower complexity.

\begin{remark}[Cholesky versus explicit matrix inversion]
For the symmetric positive definite (SPD) linear system, $\mathbf{F}\mathbf{x}=\mathbf{b}$,
we avoid explicitly forming $(\mathbf{H}+\frac{\rho}{2}\mathbf{I})^{-1}$ and instead compute a Cholesky factorization
$\mathbf{F}=\mathbf{U}^{T}\mathbf{U}$, followed by two triangular solves.
For a dense $n+1$-dimension matrix, Cholesky decomposition costs about $\tfrac{1}{3}(n+1)^3$ flops and each solve costs
$\mathcal{O}((n+1)^2)$ operations.
In contrast, the matrix inverse requires $\mathcal{O}((n+1)^3)$ cost on top of the Cholesky factorization,
making it typically at least $\sim 3$ times more expensive in leading-order work and numerically less stable \cite{nicholas2009cholesky}. 
\end{remark}

\begin{remark}[Decomposition-based implementation]
The matrix $\mathbf{H}+\frac{\rho}{2}\mathbf{I}$ is constant when $\rho$ is fixed.
Hence, we can factorize it only once and reuse the factorization in all ADMM iterations.
In particular, we employ a cached Cholesky-based solver, so that each $\bf r$-update amounts to
two triangular solves (with a possible fill-reducing permutation in the sparse case), rather
than explicitly forming matrix inverse.
In MATLAB, this is implemented via
\begin{align}
&{\mathbf F}=\texttt{decomposition}(\mathbf{H}+\frac{\rho}{2}\mathbf{I},\texttt{'chol'}), \label{DC_ADMM1}\\
&\mathbf r^{k+1}=\mathbf F\oslash\big(-{\mathbf g}+\frac{\rho}{2}(\mathbf{s} ^k+\mathbf{v}^k)\big), \label{DC_ADMM2}
\end{align}
where $\oslash$ denotes the left matrix division which actually performs the underlying triangular solves in \eqref{LU_ADMM_update} and \eqref{LU_ADMM_update2} efficiently. 
\end{remark}

\begin{algorithm}[t]
  \caption{Global CD-ADMM Algorithm to Solve \eqref{problem_ADMM}}
  \label{alg:ADMM2}
  \begin{algorithmic}[1]
    \STATE \textbf{Input:} Penalty parameter $\rho>0$, tolerance $\epsilon>0$, maximum iterations $K_{\max}$,
    initial $\mathbf{s}^{(0)}$ with $\|\mathbf{s}^{(0)}\|_2=1$, and initial scaled dual variable $\mathbf{v}^{(0)}$.
    \STATE \textbf{Output:} $\mathbf{r}^{\star}=\mathbf{r}^{(k)}$, $\mathbf{s}^{\star}=\mathbf{s}^{(k)}$, and $\mathbf{v}^{\star}=\mathbf{v}^{(k)}$.
    \STATE Pre-compute the Cholesky decomposition of $\mathbf{H}+\frac{\rho}{2}\mathbf{I}$ as \eqref{DC_ADMM1}.
    \STATE Set $k=0$;
    \REPEAT
      \STATE \textbf{Update $\mathbf{r}^{(k+1)}$} by \eqref{DC_ADMM2}:
      \STATE \hspace{1.6em}  $\mathbf r^{k+1}=\mathbf F\oslash\big(-\mathbf g+\frac{\rho}{2}(\mathbf{s} ^k+\mathbf{v}^k)\big)$;      
      \STATE \textbf{Update $\mathbf{s}^{(k+1)}$} by \eqref{ADMM_solution2} (projection onto the unit sphere):
      \STATE \hspace{1.6em} $\mathbf{s}^{(k+1)} \leftarrow \dfrac{\mathbf{r}^{(k+1)}-\mathbf{v}^{(k)}}{\left\|\mathbf{r}^{(k+1)}-\mathbf{v}^{(k)}\right\|_2}$;
      \STATE \textbf{Update $\mathbf{v}^{(k+1)}$} by \eqref{ADMM_solution3}:
      \STATE \hspace{1.6em} $\mathbf{v}^{(k+1)} \leftarrow \mathbf{v}^{(k)}+\mathbf{s}^{(k+1)}-\mathbf{r}^{(k+1)}$;
      \STATE Compute primal residual:
      \STATE \hspace{1.6em} $r_{\mathrm{pri}}^{(k+1)} \leftarrow \left\|\mathbf{s}^{(k+1)}-\mathbf{r}^{(k+1)}\right\|_2$;
      \STATE Compute dual residual:
      \STATE \hspace{1.6em} $r_{\mathrm{dual}}^{(k+1)} \leftarrow \rho\left\|\mathbf{s}^{(k+1)}-\mathbf{s}^{(k)}\right\|_2$;
      \STATE $k \leftarrow k+1$;
    \UNTIL{$\max\{r_{\mathrm{pri}}^{(k)},\,r_{\mathrm{dual}}^{(k)}\}\le \epsilon$ or $k\ge K_{\max}$}
    \STATE \textbf{return} $\mathbf{r}^{(k)}$, $\mathbf{s}^{(k)}$, and $\mathbf{v}^{(k)}$;
  \end{algorithmic}
\end{algorithm}

We summarize the algorithm for solving SCLS in \eqref{problem_ADMM} via the low-complexity Cholesky-decomposition ADMM (CD-ADMM) in Algorithm~\ref{alg:ADMM2}. 

\section{Experiments}
\label{sec:experiments}
In this section, we conduct comprehensive experiments to evaluate the performance of our proposed ADMM and CD-ADMM algorithms. Our primary goal is to demonstrate that the proposed methods achieve global optimality comparable to SOCP while offering orders-of-magnitude improvements in computational efficiency, particularly in high-dimensional and sparse settings.
\subsection{Experiments Setup}
Following the experimental protocol established in \cite{wang2022solving}, we evaluate our methods on both real-world and synthetic datasets. \\
\textbf{Real-world Dataset:}
We evaluate these methods on three publicly available datasets: Wine Quality \cite{Cortez2009ModelingWP}, residential building \cite{Rafiei2016ANM}, Insurance\footnote{https://www.kaggle.com/mirichoi0218/insurance} and BlogFeedback\footnote{https://archive.ics.uci.edu/ml/datasets/BlogFeedback} . The dataset details are provided in Appendix \ref{app:real-world}. For each dataset, we simulate the strategic data provider's target labels $z$ under different manipulation intensities (denoted as $\mathcal{A}_{\mathrm{modest}}$ and $\mathcal{A}_{\mathrm{severe}}$ ), following the settings in \cite{wang2022solving}.\\
\textbf{Synthetic Datasets}: To assess scalability, we generate synthetic data matrices $X \in \mathbb{R}^{m \times n}$ with varying dimensions ($n$ ranging from $1,000$ to $30,000$) and sparsity levels (density list as  [$0.01$, $0.001, $$0.0001$]). The labels are generated via a linear model with Gaussian noise. The synthetic dataset was generated and set via a rule from \cite{wang2021fast}. We use the code from \cite{wang2022solving} to generate these data.\\
\textbf{Baselines and Proposed Method:} We compare our proposed algorithms against three representative state-of-the-art solvers for the SCLS problem. First, to establish a ground truth for global optimality, we evaluate against two convex reformulation approaches: Semidefinite Programming (SDP) and Second-Order Cone Programming (SOCP). Both approaches are solved via interior-point methods (using MOSEK or SeDuMi). Although they yield rigorous global optimal solutions, they are computationally expensive due to the requirement of spectral decomposition. Second, to assess computational efficiency, we compare against LTRSR, a nested Lanczos method for the Trust-Region Subproblem, which represents the state-of-the-art among iterative solvers.\\
\textbf{Implementation: }Following the protocol in \cite{wang2022solving}, we evaluate the performance of different algorithms against the SDP and SOCP approaches. All results are averaged over 10 trials. We solve each reformulation to its default precision and recover the corresponding SPG-LS solutions accordingly. All real-world dataset experiments are performed on an Apple laptop equipped with an Intel Core i9 CPU (2.3 GHz) and 16 GB RAM, running MATLAB 2025b. To accommodate computational constraints, all synthetic data experiments are conducted on an Apple laptop equipped with an M4 chip and 128 GB RAM, running the same version of MATLAB.

\subsection{Results}
In this section, we present the numerical performance of our proposed algorithms. We evaluate both the solution quality and computational efficiency (runtime) of the proposed CD-ADMM compared to state-of-the-art baselines. The experiments are conducted on both real-world datasets to verify global optimality and large-scale synthetic datasets to demonstrate scalability.
\begin{table}[htbp]
\centering
\caption{Relative error of objective values}
\label{tab:rel_error_obj}

\resizebox{0.49\textwidth}{!}{%
\begin{tabular}{l ccc ccc ccc}
\toprule
\multirow{2}{*}{Dataset} &
\multicolumn{3}{c}{$(f_{\mathrm{SOCP}}-f_{\mathrm{LTRSR}})/|f_{\mathrm{SOCP}}|$} &
\multicolumn{3}{c}{$(f_{\mathrm{SOCP}}-f_{\mathrm{ADMM}})/|f_{\mathrm{SOCP}}|$} &\\
\cmidrule(lr){2-4}\cmidrule(lr){5-7}\cmidrule(lr){8-10}
& AVG & MIN & MAX  & AVG & MIN & MAX  \\
\midrule
Wine Modest &
7.10E-10 & 2.44E-12 & 5.23E-09 &

8.31E-09 & 5.08E-12 & 6.45E-08 &\\
Wine Severe &
1.12E-10 & 1.62E-12 & 1.03E-09 &

1.12E-10 & 1.65E-12 & 1.03E-09 &\\
Build Modest &
4.93E-07 & 4.91E-08 & 1.08E-06 &

5.15E-07 & 5.22E-08 & 1.09E-06 &\\
Build Severe &
3.79E-08 & 2.22E-09 & 8.09E-08 &

3.80E-08 & 2.22E-09 & 8.10E-08 & \\
Insurance Modest &
1.35E-05 & 1.02E-06 & 4.37E-05 &

2.54E-05 & 2.08E-06 & 8.11E-05 & \\
Insurance Severe &
2.57E-06 & 1.34E-07 & 7.18E-06 &
1.22E-05 & 4.45E-07 & 6.59E-05 &\\
\bottomrule
\end{tabular}%
} 
\end{table}
\subsubsection{Performance on Real-world Data}
We first validate the global optimality and convergence behavior of our method on the real-world datasets. Before analyzing the errors in detail,
Figure \ref{fig:building} illustrates the runtime comparison on the Residential Building dataset. The yellow and purple dashed lines represent the Single SDP and SOCP baselines, respectively, while the solid lines represent the iterative solvers.
As observed, both ADMM-based methods (Green and Red lines) significantly outperform the convex relaxation approaches (SDP and SOCP) by several orders of magnitude. Specifically, our CD-ADMM consistently achieves low runtime across all sample sizes ($m$). While the LTRSR method (Blue line) is slightly faster on this specific low-dimensional dataset due to its matrix-free nature, CD-ADMM remains highly competitive and is substantially faster than the standard ADMM, demonstrating the effectiveness of the Cholesky pre-computation strategy.
Besides the speed of the algorithms, Table \ref{tab:rel_error_obj} reports the relative error of the objective values obtained by CD-ADMM compared to the SOCP solver, which serves as the ground truth for global optimality. The relative errors are consistently negligible, ranging from $10^{-5}$ to $10^{-12}$ across all datasets and manipulation scenarios.  It achieves the same order of magnitude accuracy as LTRSR. This empirically verifies that our nonconvex ADMM formulation converges to the global optimum, matching the rigorous SOCP benchmark. It is worth mentioning that CD-ADMM and ADMM obtained the same experimental results, which are therefore omitted from the table. Furthermore, we visualize the objective function value against the number of iterations on the Wine Modest dataset in Figure \ref{fig:ADMM_conve}. The black dashed line represents the global optimal value computed by the SOCP baseline. We observe these main points in the figure:
\begin{itemize}
    \item \textbf{Global Convergence:} Regardless of the penalty parameter $\rho$, the objective values of our algorithm decrease and converge to the SOCP optimal value (the black dashed line).
    \item \textbf{Fast Convergence:} The algorithm demonstrates high efficiency. For moderate and large penalty parameters (e.g., $\rho \in [5, 50]$), the objective value drops sharply and reaches the vicinity of the optimal solution within the first 20 iterations.
    \item \textbf{Robustness to Parameter $\rho$:} While smaller penalty parameters (e.g., $\rho=0.5$) lead to a relatively slower decay rate, the algorithm remains stable and eventually converges. The zoomed-in sub-plot (bottom-right corner) further confirms that after 150 iterations, the discrepancy between our solution and the SOCP ground truth is virtually indistinguishable across all $\rho$ settings.
\end{itemize}
In summary, by selecting the penalty parameter $\rho$ according to the specific data characteristics, the proposed method can achieve both high computational efficiency and superior solution quality.
\begin{figure}[t]
    \centering
    \begin{subfigure}[b]{0.24\textwidth}
        \centering
        \includegraphics[width=\linewidth]{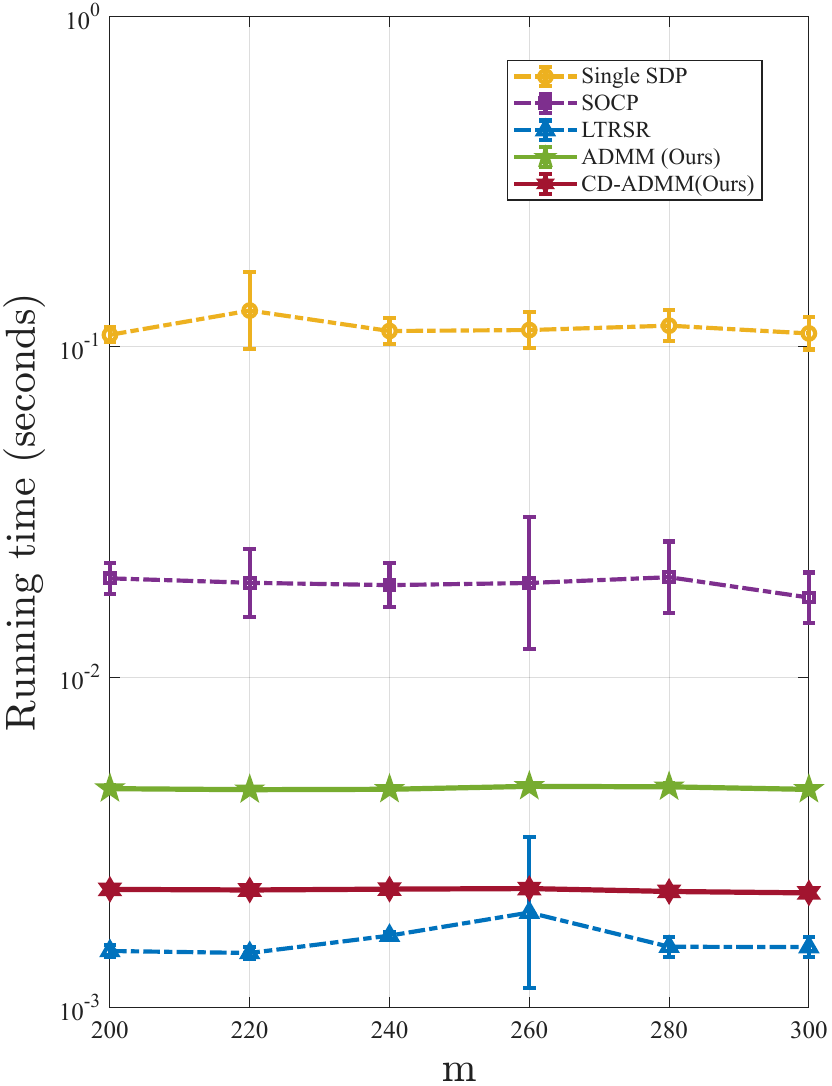}
        
    \end{subfigure}\hfill
    \begin{subfigure}[b]{0.24\textwidth}
        \centering
        \includegraphics[width=\linewidth]{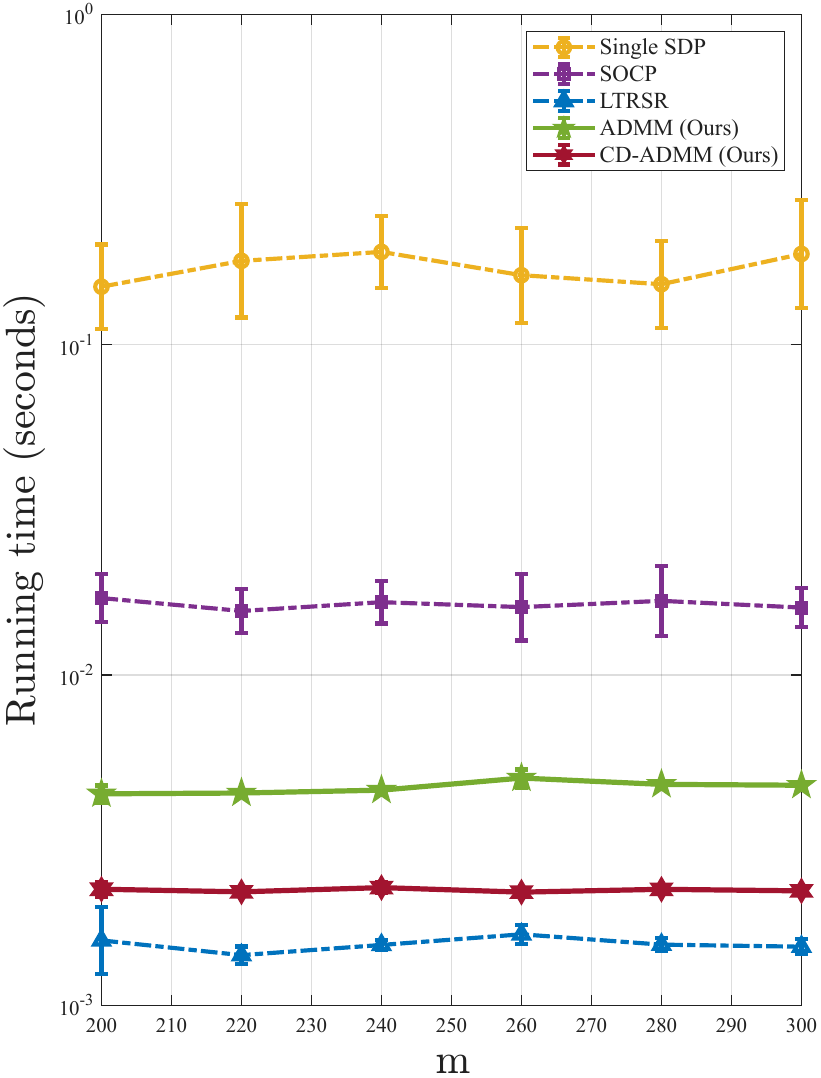}
    \end{subfigure}

    \caption{Comparison of different algorithms on the building dataset. The left and right plots correspond to $\mathcal{A}_{\mathrm{modest}}$ and $\mathcal{A}_{\mathrm{severe}}$, respectively.}
    \label{fig:building}
\end{figure}

\subsubsection{Performance on Synthetic Data}

\textbf{Dense Data:} In the dense data setting, we report the computational time in Table \ref{tab:time_dense_selected}. Regarding solution quality, our proposed methods achieve global optimality comparable to the state-of-the-art baselines across all dense instances; the detailed relative error results are deferred to Appendix \ref{app:dense}.
In terms of computational efficiency, the SOCP approach exhibits a heavy dependency on the spectral decomposition (recorded as ``eig time" in parentheses), which consistently consumes a substantial portion of the total runtime. In contrast, our proposed ADMM and CD-ADMM effectively bypass this bottleneck. As evidenced by Table 2, our methods consistently outperform the SOCP solver. For example, in the case of $m=2n$ and $n=10,000$, CD-ADMM achieves a speedup of approximately 1.89 times compared to SOCP. Furthermore, our approach remains highly competitive with the leading iterative solver, LTRSR. While LTRSR performs efficiently in lower dimensions, our CD-ADMM demonstrates robust scalability, even surpassing LTRSR in specific high-dimensional regimes. This confirms that our method offers a favorable trade-off between global optimality guarantees and runtime efficiency on dense data.

\textbf{Sparse Data:} We proceed to evaluate the performance on sparse synthetic datasets (sparsity $= 0.0001$). According to Table \ref{tab:sparsity_0001}, the SOCP approach is significantly slower across all regimes due to its reliance on spectral decomposition, which cannot fully exploit data sparsity; consequently, our CD-ADMM method achieves substantial speedups, reaching approximately 500 times faster than SOCP in the case of $m=0.5n$ and $n=30,000$. When compared with the iterative solver LTRSR, our method exhibits a distinct advantage in high-dimensional settings ($m \le n$): for instance, with $m=0.5n$ and $n=30,000$, CD-ADMM converges in 0.684s, outperforming LTRSR's 1.347s, which highlights the efficiency of the pre-computed sparse Cholesky factorization when the system size is dominated by the feature dimension. However, we observe that in large-sample regimes ($m > n$) with extremely large dimensions ($n \ge 25,000$), LTRSR exhibits better scalability due to its matrix-free nature; for example, at $m=3n$ and $n=30,000$, LTRSR (2.779s) surpasses CD-ADMM (8.972s) as the computational overhead of triangular solves increases in the over-determined regime.  Overall, CD-ADMM provides the best performance for high-dimensional problems ($m \le n$) and remains highly competitive (and orders of magnitude faster than SOCP) even in large-sample regimes, making it a robust choice for diverse sparse adversarial learning tasks.For complete experimental results on datasets with varying sparsity levels, please refer to Appendix \ref{app:sparse}. 

\begin{table}[htbp]
\centering
\caption{Time (s) on Synthetic Data without Sparsity ($\gamma = 0.1$) Ratio1 $=\mathrm{SOCP}/\mathrm{CD\text{-}ADMM}$ and Ratio2 $=\mathrm{ADMM}/\mathrm{CD\text{-}ADMM}$}
\label{tab:time_dense_selected}
\resizebox{0.48\textwidth}{!}{%
\begin{tabular}{rcccccc}
\toprule

\multicolumn{7}{c}{$m = 2n$} \\
\midrule
$n$ & SOCP (eig time) & LTRSR & ADMM & CD-ADMM & Ratio1  & Ratio2 \\
\midrule
1000  & 0.167 (0.087) & 0.101 & \textbf{0.027} & 0.071 & 2.35 & 0.26 \\
2000  & 0.591 (0.343) & 0.440 & \textbf{0.164} & 0.287 & 2.06 & 0.37 \\
4000  & 3.099 (1.677) & 1.962 & \textbf{1.280} & 1.852 & 1.67 & 0.65 \\
6000  & 8.461 (3.629) & 4.495 & \textbf{4.436} & 4.691 & 1.80 & 0.99 \\
8000  & 17.017 (7.130) & \textbf{8.942} & 9.696 & 9.488 & 1.79 & 1.08 \\
10000 & 31.149 (12.505) & \textbf{16.118} & 18.571 & 16.518 & 1.89 & 1.15 \\

\midrule
\multicolumn{7}{c}{$m = n$} \\
\midrule
$n$ & SOCP (eig time) & LTRSR & ADMM & CD-ADMM & Ratio1  & Ratio2 \\
\midrule
1000  & 0.144 (0.087) & 0.096 & \textbf{0.035} & 0.076 & 1.90 & 0.36 \\
2000  & 0.555 (0.340) & 0.387 & \textbf{0.142} & 0.272 & 2.04 & 0.37 \\
4000  & 2.739 (1.582) & 1.630 & \textbf{1.100} & 1.583 & 1.73 & 0.67 \\
6000  & 7.438 (3.602) & 3.859 & \textbf{3.657} & 4.079 & 1.82 & 0.95 \\
8000  & 15.739 (6.856) & 9.063 & 8.354 & \textbf{8.175} & 1.93 & 0.92 \\
10000 & 30.129 (12.431) & 14.039 & 15.937 & \textbf{13.831} & 2.18 & 1.14 \\

\midrule
\multicolumn{7}{c}{$m = 0.5n$} \\
\midrule
$n$ & SOCP (eig time) & LTRSR & ADMM & CD-ADMM & Ratio1  & Ratio2 \\
\midrule
1000  & 0.154 (0.099) & 0.105 & \textbf{0.038} & 0.080 & 1.91 & 0.36 \\
2000  & 0.527 (0.328) & 0.230 & \textbf{0.136} & 0.262 & 2.01 & 0.59 \\
4000  & 2.680 (1.593) & \textbf{0.746} & 1.036 & 1.545 & 1.73 & 1.39 \\
6000  & 6.982 (3.300) & \textbf{2.334} & 3.374 & 3.855 & 1.81 & 1.45 \\
8000  & 14.687 (6.409) & \textbf{3.998} & 7.732 & 7.539 & 1.95 & 1.93 \\
10000 & 27.205 (11.090) & \textbf{4.939} & 14.614 & 12.516 & 2.17 & 2.96 \\

\bottomrule
\end{tabular}%
}
\end{table}

\begin{table}[t]
\centering
\caption{Time (seconds) on synthetic data (sparsity $=0.0001$). Ratio1 $=\mathrm{SOCP}/\mathrm{CD\text{-}ADMM}$ and
Ratio2 $=\mathrm{ADMM}/\mathrm{CD\text{-}ADMM}$.
The fastest time in each row is in bold. $\gamma = 0.1$.}
\label{tab:sparsity_0001}
\setlength{\tabcolsep}{6pt}
\resizebox{0.48\textwidth}{!}{%
\begin{tabular}{c r l r r r r r}
\hline
$m$ & $n$ & SOCP (eig) & LTRSR & ADMM & CD-ADMM & Ratio1 & Ratio2 \\
\hline
$0.5n$ & 10000 & 12.196 (11.623) & 0.653 & 13.390 & \textbf{0.085} & 143 & 157 \\
$0.5n$ & 15000 & 34.800 (33.016) & 0.638 & 42.199 & \textbf{0.116} & 301 & 365 \\
$0.5n$ & 20000 & 79.709 (72.396) & 0.842 & 100.689 & \textbf{0.310} & 257 & 325 \\
$0.5n$ & 25000 & 186.674 (165.716) & 1.090 & 186.911 & \textbf{0.442} & 423 & 423 \\
$0.5n$ & 30000 & 342.463 (293.426) & 1.347 & 318.067 & \textbf{0.684} & 501 & 465 \\
\hline
$n$  & 10000 & 14.041 (12.492) & 0.502 & 13.749 & \textbf{0.049} & 285 & 279 \\
$n$  & 15000 & 47.203 (41.625) & 0.757 & 43.796 & \textbf{0.179} & 264 & 245 \\
$n$  & 20000 & 118.122 (96.957) & 0.976 & 100.073 & \textbf{0.377} & 313 & 266 \\
$n$  & 25000 & 236.470 (188.916) & 1.266 & 191.305 & \textbf{0.829} & 285 & 231 \\
$n$  & 30000 & 432.633 (333.367) & \textbf{1.694} & 88.595 & 2.006 & 216 & 44 \\
\hline
$2n$ & 10000 & 16.205 (13.393) & 0.701 & 13.726 & \textbf{0.101} & 160 & 136 \\
$2n$ & 15000 & 55.850 (45.386) & 0.946 & 12.803 & \textbf{0.259} & 215 & 49 \\
$2n$ & 20000 & 142.832 (109.946) & 1.315 & 27.696 & \textbf{0.811} & 176 & 34 \\
$2n$ & 25000 & 281.458 (209.865) & \textbf{1.720} & 53.666 & 2.568 & 110 & 21 \\
$2n$ & 30000 & 481.532 (353.768) & \textbf{2.227} & 90.885 & 5.790 & 83 & 16 \\
\hline
$3n$ & 10000 & 18.850 (14.842) & 0.810 & 13.647 & \textbf{0.109} & 173 & 125 \\
$3n$ & 15000 & 61.937 (47.430) & 1.204 & 12.798 & \textbf{0.434} & 143 & 29 \\
$3n$ & 20000 & 150.974 (112.789) & 1.619 & 28.738 & \textbf{1.510} & 100 & 19 \\
$3n$ & 25000 & 295.949 (218.266) & \textbf{2.118} & 54.505 & 3.837 & 77 & 14 \\
$3n$ & 30000 & 489.773 (343.090) & \textbf{2.779} & 83.778 & 8.972 & 55 & 9 \\
\hline
\end{tabular}}
\end{table}

\begin{figure}[t]
\centering{\includegraphics[width=0.98\columnwidth]{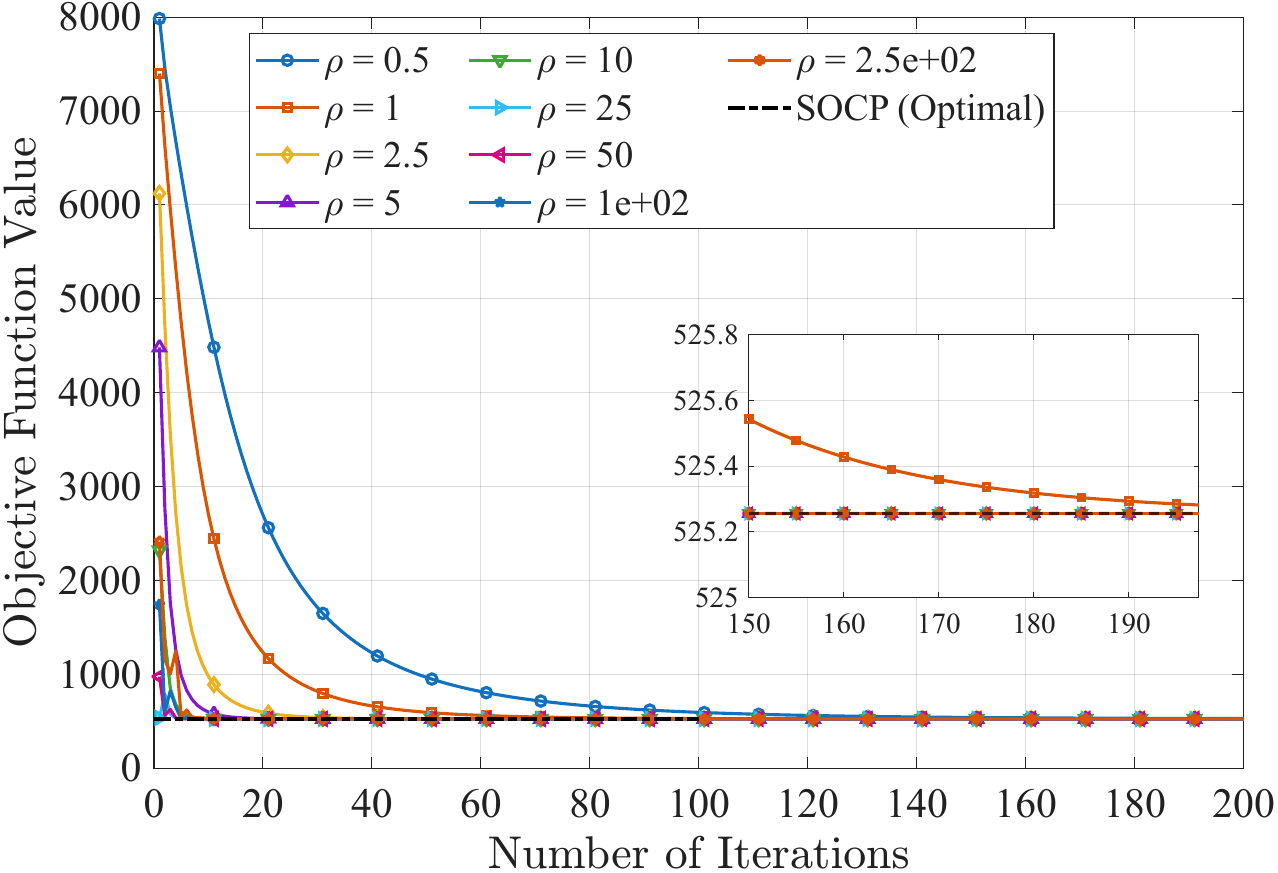}}
\caption{Convergence Curves on the Wine Modest for Different Values of $\rho$}
\label{fig:ADMM_conve}
\end{figure}

\section{Conclusion}
We studied the Stackelberg prediction game with least-squares losses and leveraged an exact SCLS reformulation to obtain a tractable single-level model. Building on a consensus splitting of the SCLS problem, we developed a low-complexity ADMM solver with closed-form updates. This structure enables scalable implementations via one-time Cholesky decomposition to avoid explicit matrix inversion. We further established a duality-based characterization and showed that the proposed iterations converge to a primal-dual point satisfying the optimality conditions of the SCLS problem. Experiments on synthetic and real datasets demonstrated that our method achieves competitive solution quality while substantially reducing runtime compared to state-of-the-art SCLS solvers in large-scale and sparse datasets.

\paragraph{Limitations and Future Works.}
Although our method involves only a single tuning parameter $\rho$, its practical performance for different datasets may still depend on this choice. 
A promising direction is to develop adaptive penalty-update rules that are both computationally light and equipped with rigorous convergence guarantees, and to incorporate Bayesian optimization for automatic, data-driven hyperparameter tuning. 
It is also of interest to extend the proposed framework beyond least-squares losses and linear predictors, such as fractional quadratic losses and highly nonlinear predictors, thereby covering broader classes of Stackelberg learning problems.

\bibliography{example_paper}
\bibliographystyle{icml2026}

\newpage
\appendix
\onecolumn
\section{Proof of Strong Duality}
\label{proof:psd}
\begin{proof}
The Lagrangian of problem~\eqref{problem_ADMM} is
\begin{equation}\label{eq:SCLS_Lagrangian}
\mathcal{L}_0(\mathbf{r},\mathbf{s},\mathbf{u},\nu)
=
f(\mathbf{r})+\nu(\mathbf{s}^T\mathbf{s}-1)+\mathbf{u}^T(\mathbf{s}-\mathbf{r}),
\end{equation}
where $\nu\in\mathbb{R}$ and $\mathbf{u}\in\mathbb{R}^{n+1}$ are the dual variables. Since the constraint
$\mathbf{s}^T\mathbf{s}=1$ must be active at optimality, we focus on the non-trivial case where $\nu\neq 0$ and $\mathbf{u}\neq \mathbf{0}$.

Setting the gradients of~\eqref{eq:SCLS_Lagrangian} with respect to $\mathbf{r}$ and $\mathbf{s}$ to zero gives
\begin{subequations}\label{eq:SCLS_Lag_der}
\begin{align}
\nabla_{\mathbf{r}}\mathcal{L}_0(\mathbf{r},\mathbf{s},\mathbf{u},\nu)
&=
\nabla f(\mathbf{r})-\mathbf{u}
=
2\mathbf{H}\mathbf{r}+2\mathbf{g}-\mathbf{u}
=\mathbf{0}, \label{eq:SCLS_Lag_der_r}\\
\nabla_{\mathbf{s}}\mathcal{L}_0(\mathbf{r},\mathbf{s},\mathbf{u},\nu)
&=
2\nu\mathbf{s}+\mathbf{u}
=\mathbf{0}. \label{eq:SCLS_Lag_der_s}
\end{align}
\end{subequations}
Therefore, the minimizing primal variables for fixed $(\mathbf{u},\nu)$ satisfy
\begin{subequations}\label{eq:SCLS_Lag_opt}
\begin{align}
\mathbf{r}^\star
&=
\mathbf{H}^{-1}\!\left(\frac{1}{2}\mathbf{u}-\mathbf{g}\right), \label{eq:SCLS_Lag_opt_r}\\
\mathbf{s}^\star
&=
-\frac{1}{2\nu}\mathbf{u}. \label{eq:SCLS_Lag_opt_s}
\end{align}
\end{subequations}

Substituting~\eqref{eq:SCLS_Lag_opt} into~\eqref{eq:SCLS_Lagrangian} defines the dual function
$g(\nu,\mathbf{u})=\inf_{\mathbf{r},\mathbf{s}}\mathcal{L}_0(\mathbf{r},\mathbf{s},\mathbf{u},\nu)$, which can be given by
\begin{align} \label{EQCQP_ADMM_dual_obj}
    g(\nu,{\bf u}) 
  = & \inf_{\mathbf{r},\mathbf{s}}\mathcal{L}_0(\mathbf{r},\mathbf{s},{\bf u},\nu) = \mathcal{L}_0({\bf r}^{\star},{\bf s}^{\star},{\bf u},\nu) \nonumber \\
  = &  
  p- {\bf g}^T{\bf H}^{-1} {\bf g}^T
  -\frac{1}{4\nu}{\bf u}^T{\bf u}-{\nu}-\frac{1}{4}{\bf u}^T{\bf H}^{-1}{\bf u}-
  {\bf u}^T{\bf H}^{-1}{\bf g}. 
\end{align}
Then, the dual problem is
\begin{equation}\label{eq:SCLS_dual}
d^\star:\ \max_{\nu,\mathbf{u}}\ g(\nu,\mathbf{u}).
\end{equation}

Next, we characterize the dual optimizer by stationarity of $g(\nu,\mathbf{u})$.
From~\eqref{eq:SCLS_Lag_opt_s} and the primal feasibility $\|\mathbf{s}^\star\|_2^2=1$, we obtain
\begin{equation}\label{eq:SCLS_nu_relation}
\|\mathbf{s}^\star\|_2^2
=
\frac{\|\mathbf{u}\|_2^2}{4\nu^2}
=1
\quad\Longrightarrow\quad
\nu^\star=\pm\frac{\|\mathbf{u}^\star\|_2}{2}.
\end{equation}
Moreover, combining~\eqref{eq:SCLS_Lag_opt_r} with the consensus constraint $\mathbf{s}^\star=\mathbf{r}^\star$ yields
\begin{equation}\label{eq:SCLS_consensus_relation}
-\frac{1}{2\nu^\star}\mathbf{u}^\star
=
\mathbf{H}^{-1}\!\left(\frac{1}{2}\mathbf{u}^\star-\mathbf{g}\right),
\end{equation}
which implies $\mathbf{s}^\star-\mathbf{r}^\star=\mathbf{0}$ at the saddle point. Consequently, the linear term
$(\mathbf{u}^\star)^T(\mathbf{s}^\star-\mathbf{r}^\star)$ in~\eqref{eq:SCLS_Lagrangian} vanishes.

Combining \eqref{eq:SCLS_Lag_opt}, \eqref{eq:SCLS_nu_relation}, and \eqref{eq:SCLS_consensus_relation}, i.e.,
at $(\nu^\star,\mathbf{u}^\star)$ and the corresponding $(\mathbf{r}^\star,\mathbf{s}^\star)$, we then have
\begin{align}\label{eq:SCLS_strong_duality}
d^\star
=
g(\nu^\star,\mathbf{u}^\star)
=
\mathcal{L}_0(\mathbf{r}^\star,\mathbf{s}^\star,\mathbf{u}^\star,\nu^\star)
=
f(\mathbf{r}^\star)
=
p^\star,
\end{align}
which shows that the duality gap is zero. Hence, strong duality holds for~\eqref{problem_ADMM}, and the corresponding
Lagrangian saddle point $(\mathbf{r}^\star,\mathbf{s}^\star,\mathbf{u}^\star,\nu^\star)$ yields a globally optimal solution.
This completes the proof.
\end{proof}

\section{Proof of Global Convergence}
\label{proof:pgc}
\begin{proof}
We prove the convergence of the proposed ADMM scheme for problem~\eqref{problem_ADMM} with a fixed penalty parameter $\rho>0$.

\paragraph{Step 1: Preliminaries and residual definitions.}
Recall the scaled augmented Lagrangian
\begin{align}\label{eq:scaled_AL}
\mathcal{L}_2(\mathbf{r},\mathbf{s},\mathbf{v})
=
f(\mathbf{r})
+\frac{\rho}{2}\|\mathbf{s}-\mathbf{r}+\mathbf{v}\|_2^2
-\frac{\rho}{2}\|\mathbf{v}\|_2^2,
\end{align}
where $\mathbf{s}^T\mathbf{s}=1$.
Define the primal and dual residuals
\[
\mathbf{d}^{k+1}\triangleq \mathbf{s}^{k+1}-\mathbf{r}^{k+1},\qquad
r_{\mathrm{pri}}^{k+1}\triangleq \|\mathbf{d}^{k+1}\|_2,\qquad
r_{\mathrm{dual}}^{k+1}\triangleq \rho\|\mathbf{s}^{k+1}-\mathbf{s}^k\|_2,
\]
and note that $\mathbf{v}^{k+1}-\mathbf{v}^k=\mathbf{d}^{k+1}$.

\paragraph{Step 2: Descent property of the $\mathbf{r}$-update.}
The $\mathbf{r}$-subproblem in~\eqref{ADMM_update1} minimizes
\[
f(\mathbf{r})+\frac{\rho}{2}\|\mathbf{s}^k-\mathbf{r}+\mathbf{v}^k\|_2^2,
\]
whose Hessian equals $2\mathbf{H}+\rho\mathbf{I}\succeq \rho\mathbf{I}$.
Hence, the objective is $\rho$-strongly convex.
By the standard descent property of strongly convex functions, we obtain
\begin{equation}\label{eq:r_descent_proof}
\mathcal{L}_2(\mathbf{r}^k,\mathbf{s}^k,\mathbf{v}^k)
-\mathcal{L}_2(\mathbf{r}^{k+1},\mathbf{s}^k,\mathbf{v}^k)
\ge
\frac{\rho}{2}\|\mathbf{r}^{k+1}-\mathbf{r}^k\|_2^2.
\end{equation}

\paragraph{Step 3: Descent property of the $\mathbf{s}$-update.}
The $\mathbf{s}$-subproblem in~\eqref{ADMM_update2} is an exact projection onto the unit sphere.
Therefore,
\begin{equation}\label{eq:s_descent_proof}
\mathcal{L}_2(\mathbf{r}^{k+1},\mathbf{s}^k,\mathbf{v}^k)
-
\mathcal{L}_2(\mathbf{r}^{k+1},\mathbf{s}^{k+1},\mathbf{v}^k)
\ge 0.
\end{equation}

\paragraph{Step 4: Effect of the dual update.}
Using $\mathbf{v}^{k+1}=\mathbf{v}^k+\mathbf{d}^{k+1}$, we directly compute
\begin{equation}\label{eq:v_increase_proof}
\mathcal{L}_2(\mathbf{r}^{k+1},\mathbf{s}^{k+1},\mathbf{v}^{k+1})
=
\mathcal{L}_2(\mathbf{r}^{k+1},\mathbf{s}^{k+1},\mathbf{v}^k)
+\frac{\rho}{2}\|\mathbf{d}^{k+1}\|_2^2.
\end{equation}

\paragraph{Step 5: Lyapunov function and global descent.}
Define the Lyapunov function for $k\ge1$ as
\begin{equation}\label{eq:Lyapunov_proof}
\Phi^k
\triangleq
\mathcal{L}_2(\mathbf{r}^k,\mathbf{s}^k,\mathbf{v}^k)
+\frac{\rho}{2}\|\mathbf{s}^k-\mathbf{s}^{k-1}\|_2^2.
\end{equation}
Combining \eqref{eq:r_descent_proof}--\eqref{eq:v_increase_proof} and using
$\mathbf{s}^{k+1}-\mathbf{s}^k\to0$, we obtain
\begin{equation}\label{eq:Phi_descent_proof}
\Phi^k-\Phi^{k+1}
\ge
\frac{\rho}{2}\|\mathbf{r}^{k+1}-\mathbf{r}^k\|_2^2
+\frac{\rho}{2}\|\mathbf{s}^{k+1}-\mathbf{s}^k\|_2^2
+\frac{\rho}{2}\|\mathbf{d}^{k+1}\|_2^2.
\end{equation}
Since $f(\mathbf{r})$ is bounded below on the unit sphere, $\{\Phi^k\}$ is bounded below and monotonically decreasing, hence convergent.

\paragraph{Step 6: Residual and variable convergence.}
Summing~\eqref{eq:Phi_descent_proof} over $k$ yields
\[
\sum_{k=0}^{\infty}\|\mathbf{r}^{k+1}-\mathbf{r}^k\|_2^2
+
\sum_{k=0}^{\infty}\|\mathbf{s}^{k+1}-\mathbf{s}^k\|_2^2
+
\sum_{k=0}^{\infty}\|\mathbf{d}^{k+1}\|_2^2
<\infty.
\]
Consequently,
\[
r_{\mathrm{pri}}^k=\|\mathbf{s}^k-\mathbf{r}^k\|_2\to0,
\qquad
r_{\mathrm{dual}}^k=\rho\|\mathbf{s}^k-\mathbf{s}^{k-1}\|_2\to0,
\]
and $\|\mathbf{v}^{k+1}-\mathbf{v}^k\|_2\to0$.

\paragraph{Step 7: Stationary point convergence.}
Let $(\mathbf{r}^{k_j},\mathbf{s}^{k_j},\mathbf{v}^{k_j})$ be a convergent subsequence with limit
$(\bar{\mathbf{r}},\bar{\mathbf{s}},\bar{\mathbf{v}})$.
From $r_{\mathrm{pri}}^k\to0$ and $\mathbf{s}^{kT}\mathbf{s}^k=1$, we have
$\bar{\mathbf{r}}=\bar{\mathbf{s}}$ and $\|\bar{\mathbf{r}}\|_2=1$.
Taking limits in the optimality conditions of the $\mathbf{r}$- and $\mathbf{s}$-subproblems shows that
$(\bar{\mathbf{r}},\bar{\mathbf{s}},\bar{\mathbf{v}})$ satisfies the KKT conditions of~\eqref{problem_ADMM}.
Therefore, every accumulation point of the ADMM iterates is a stationary point of problem~\eqref{problem_ADMM}.

This completes the proof.
\end{proof}

\section{More Experimental Results}
In this section, we provide a comprehensive presentation of the experimental results. We focus on reporting the wall-clock time and the relative accuracy of various algorithms on different datasets with respect to the baseline method.
\subsection{Real-World Data}
\label{app:real-world}
For a fair comparison, we adopt the same experimental settings for the real-world datasets Wine Quality \cite{Cortez2009ModelingWP}, residential building \cite{Rafiei2016ANM}, Insurance\footnote{https://www.kaggle.com/mirichoi0218/insurance} and BlogFeedback\footnote{https://archive.ics.uci.edu/ml/datasets/BlogFeedback} as those used in \cite{wang2019global,wang2022solving}. The detailed characteristics of the datasets are summarized in Table~\ref{tab:datasets_real}. We also compare the running time of all algorithms on these datasets, as shown in Figure \ref{fig:wine_exp}, \ref{fig:insurance_exp}, and \ref{fig:blogfeedback_exp}
\begin{table}[htbp]
    \centering
    \caption{Statistics of Real-world Datasets}
    \label{tab:datasets_real}
    \resizebox{0.8\textwidth}{!}{%
    \begin{tabular}{lcccc}
        \toprule
        \textbf{Dataset} & \textbf{Samples ($m$)} & \textbf{Features ($n$)} & \textbf{Source} & \textbf{Description} \\
        \midrule
        \textbf{Wine Quality}  & 1,599 & 11 & UCI  & Physicochemical properties \\
        \textbf{Residential Building}  & 372 & 107 & UCI  & Construction \& market data \\
        \textbf{Insurance} & 1,338 & 7 & Kaggle  & Beneficiary demographics \\
        \textbf{BlogFeedback}  & 52,397 & 281 & UCI  & Blog post metadata \\
        \bottomrule
    \end{tabular}%
    }
\end{table}
\begin{figure}[t] 
    \centering
    \begin{subfigure}[b]{0.49\linewidth}
        \centering
    \includegraphics[width=\linewidth]{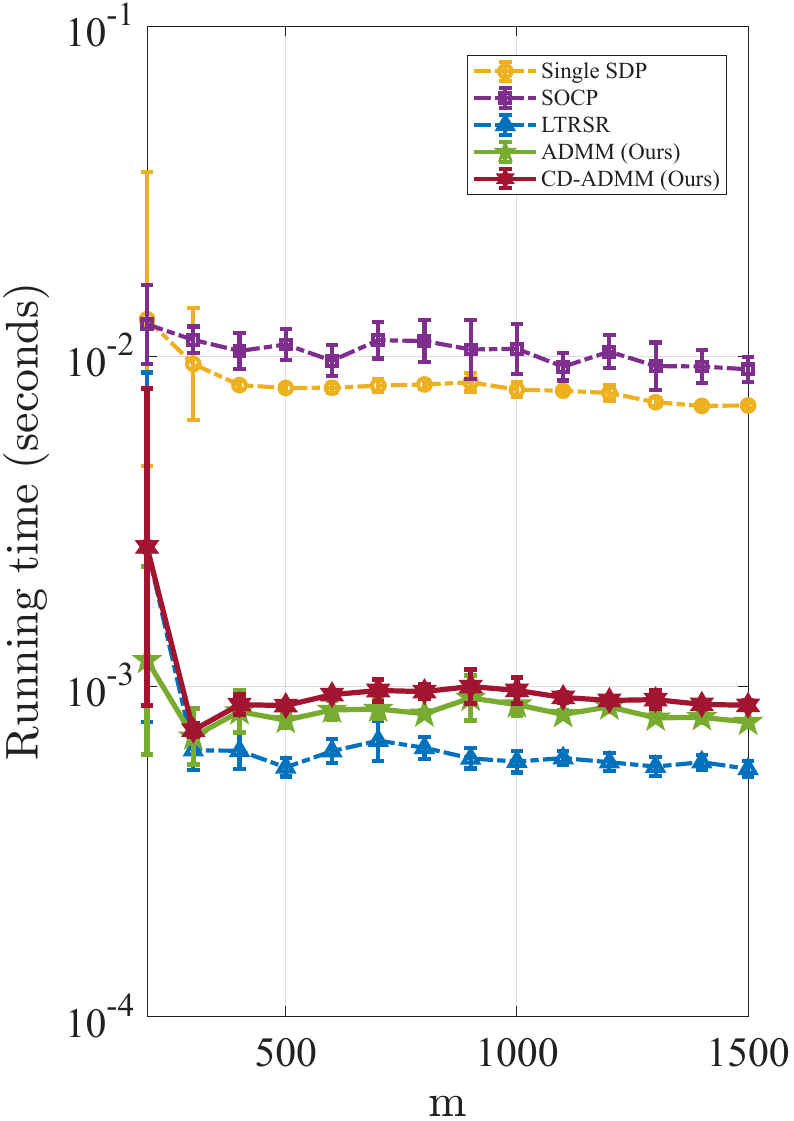}
        
    \end{subfigure}
    \hfill
    \begin{subfigure}[b]{0.49\linewidth}
        \centering        \includegraphics[width=\linewidth]{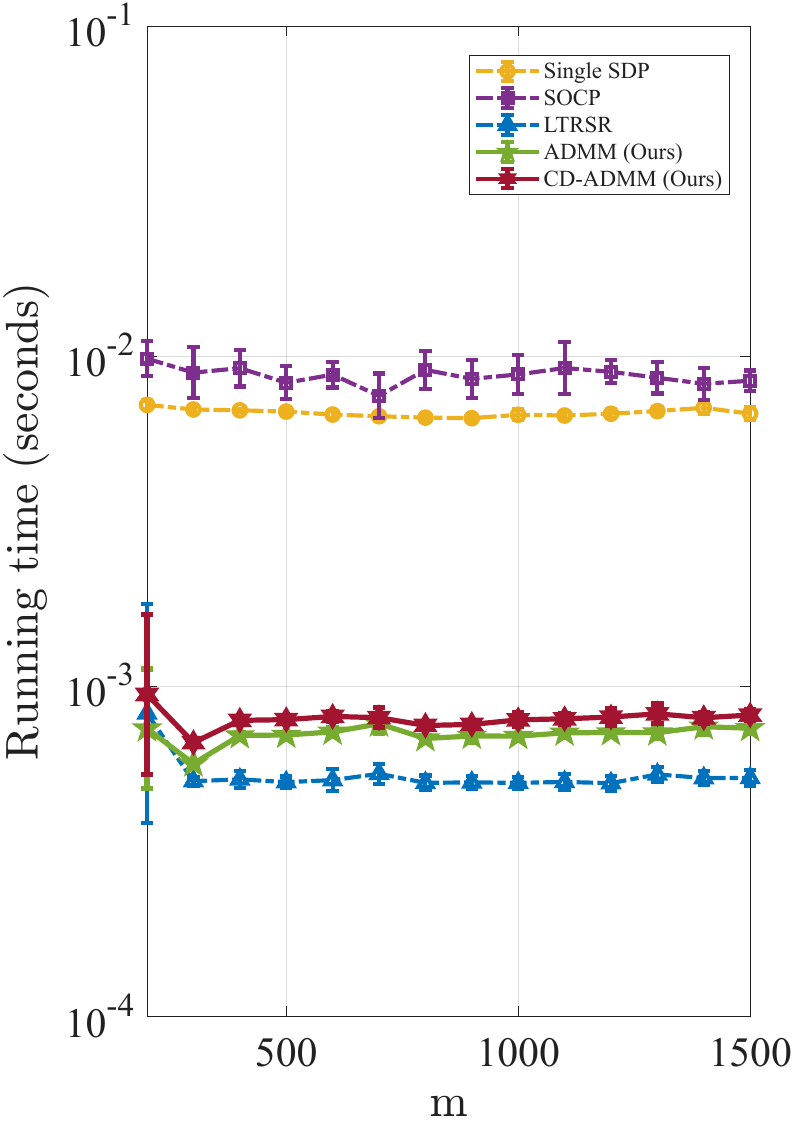} 
    \end{subfigure}

    \caption{Comparison of different algorithms on the wine dataset. The left and right plots correspond to $\mathcal{A}_{\mathrm{modest}}$ and $\mathcal{A}_{\mathrm{severe}}$, respectively.}
    \label{fig:wine_exp}
    \vspace{-0.3cm} 
\end{figure}

\begin{figure}[t] 
    \centering
    \begin{subfigure}[b]{0.49\linewidth}
        \centering
    \includegraphics[width=\linewidth]{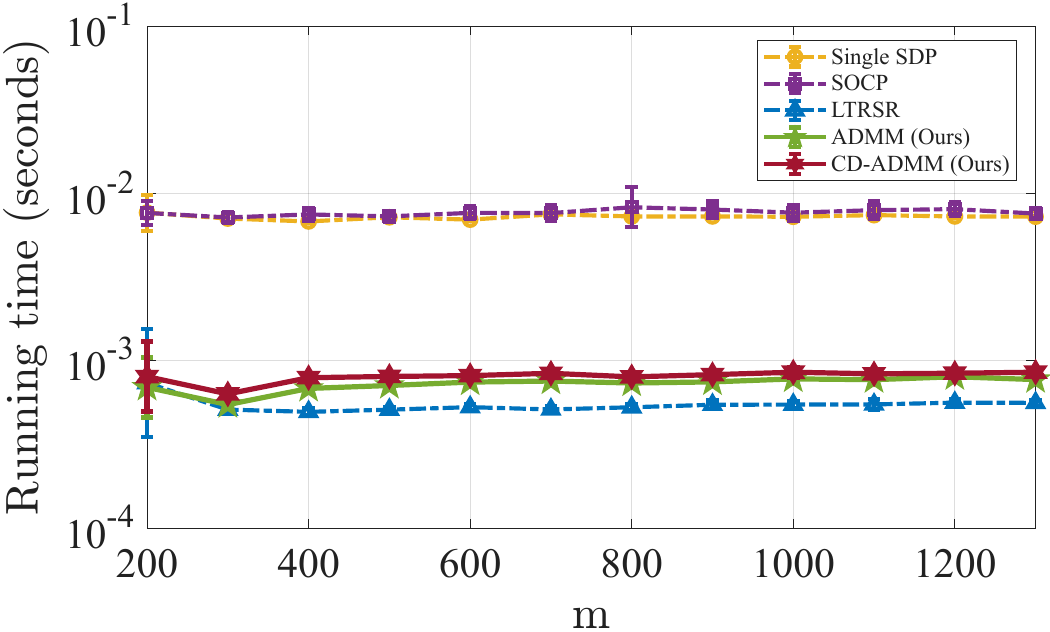}
        
    \end{subfigure}
    \hfill
    \begin{subfigure}[b]{0.49\linewidth}
        \centering        \includegraphics[width=\linewidth]{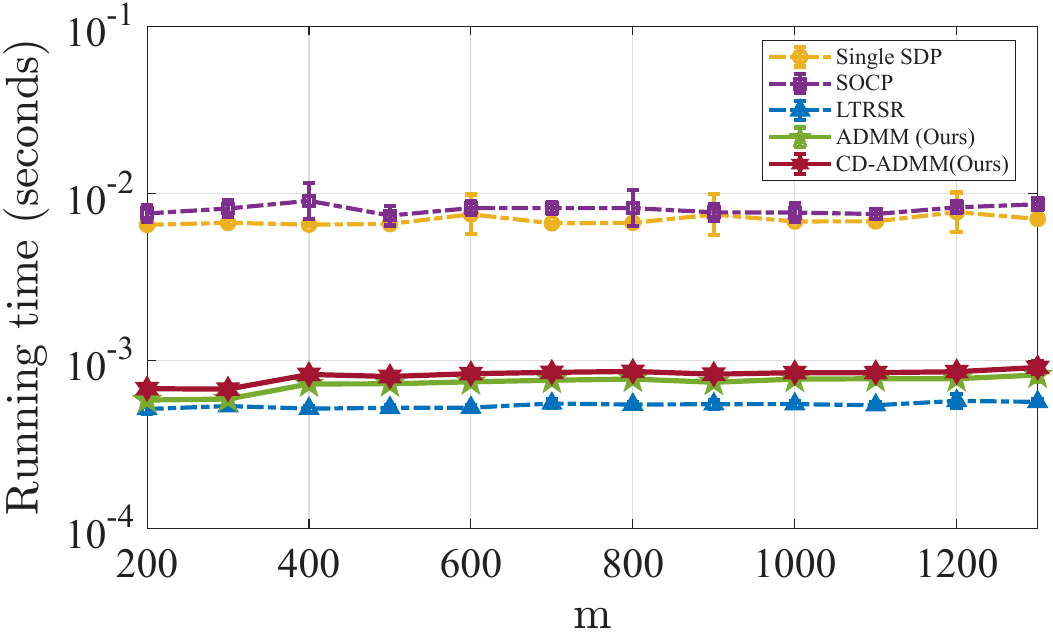} 
    \end{subfigure}

    \caption{Comparison of different algorithms on the insurance dataset. The left and right plots correspond to $\mathcal{A}_{\mathrm{modest}}$ and $\mathcal{A}_{\mathrm{severe}}$, respectively.}
    \label{fig:insurance_exp}
    \vspace{-0.3cm} 
\end{figure}

\begin{figure}[t]
    \centering
    \begin{subfigure}[b]{0.49\linewidth}
        \centering
      \includegraphics[width=\linewidth]{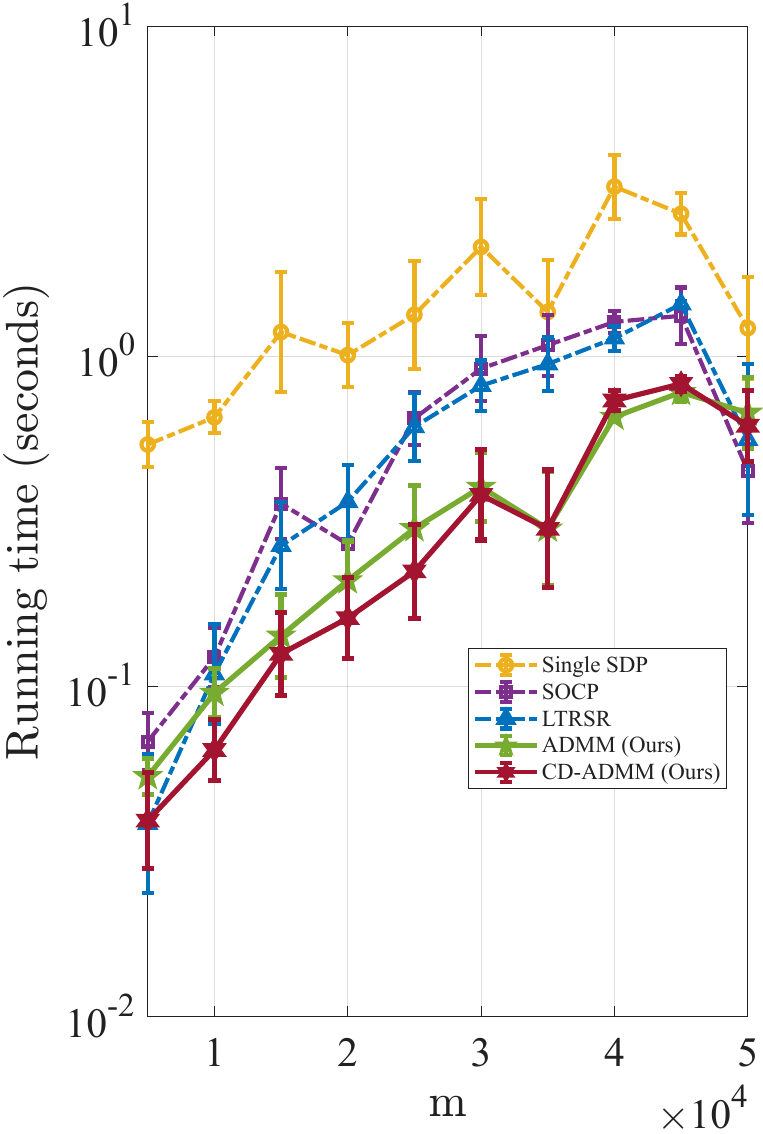}
    \end{subfigure}
    \hfill
    \begin{subfigure}[b]{0.49\linewidth}
        \centering
        \includegraphics[width=\linewidth]{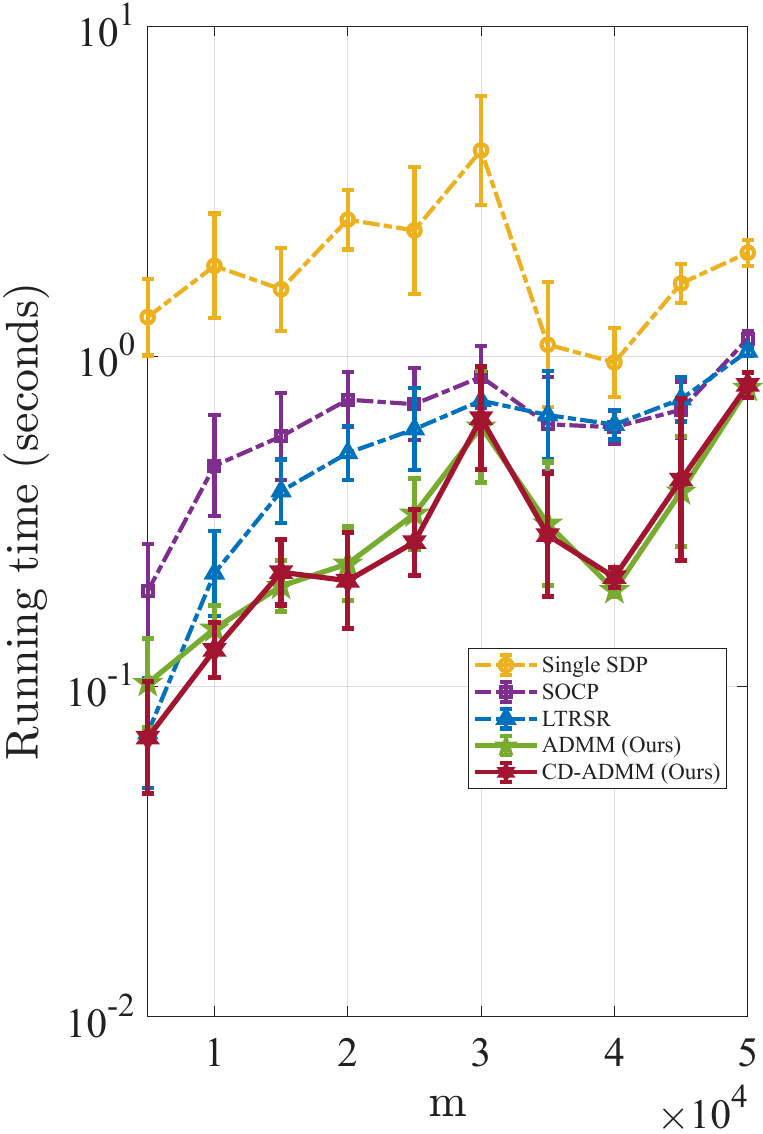}

    \end{subfigure}
    \caption{Comparison of different algorithms on the blogfeedback dataset. The left and right plots correspond to $\mathcal{A}_{\mathrm{modest}}$ and $\mathcal{A}_{\mathrm{severe}}$, respectively.}
    \label{fig:blogfeedback_exp}
    \vspace{-0.3cm}
\end{figure}

\subsection{Synthetic Dataset}
To further substantiate the scalability and efficiency of our proposed reformulation, particularly in high-dimensional regimes, we conduct comprehensive experiments on synthetic datasets across a wide spectrum of hyperparameters, dimensions, and sparsity levels.
\subsubsection{Dense Data}
\label{app:dense}
Table \ref{tab:time_dense_gamma001} reports the wall-clock time on dense synthetic dataset across sample-to-features ratios ($m/n$) with $/gamma=0.001$. The proposed ADMM-based methods provide significant speedups. Specifically, the standard ADMM is most efficient for moderate dimensions ($n \le 6000$), while CD-ADMM demonstrates superior scalability in high-dimensional regimes (e.g., $m=n, n=10,000$), outperforming both SOCP and the standard ADMM as the feature space expands.
In addition to computational efficiency, we further validate the solution quality of our proposed methods. Table \ref{tab:relerr_dense_gamma01} and \ref{tab:relerr_dense_gamma001} demonstrate that the significant speedup achieved by our algorithm does not come at the expense of accuracy, yielding precision comparable to state-of-the-art baselines.

\begin{table}[htbp]
\centering
\caption{Time (seconds) on synthetic data without sparsity ($\gamma=0.01$). Ratio1 $=\mathrm{SOCP}/\mathrm{CD\text{-}ADMM}$ and Ratio2 $=\mathrm{ADMM}/\mathrm{CD\text{-}ADMM}$ The fastest time in each row is in bold.}
\label{tab:time_dense_gamma001}
\setlength{\tabcolsep}{6pt}
\begin{tabular}{rcccccc}
\toprule

\multicolumn{7}{c}{$m = 2n$} \\
\midrule
$n$ & SOCP (eig time) & LTRSR & ADMM & CD-ADMM & Ratio1  & Ratio2 \\
\midrule
1000  & 0.2942 (0.0974) & 0.1014 & \textbf{0.0394} & 0.0649 & 4.53 & 0.39 \\
2000  & 0.7873 (0.3467) & 0.4100 & \textbf{0.1583} & 0.2809 & 2.80 & 0.39 \\
4000  & 3.6944 (1.6107) & 1.9316 & \textbf{1.2777} & 1.7890 & 2.07 & 0.66 \\
6000  & 7.7093 (3.6017) & 4.7605 & \textbf{4.2252} & 4.6706 & 1.65 & 0.89 \\
8000  & 18.1199 (6.9188) & \textbf{8.4709} & 9.9534 & 9.7542 & 1.86 & 1.18 \\
10000 & 32.0714 (12.4821) & \textbf{13.7620} & 18.8127 & 16.7652 & 1.91 & 1.37 \\

\midrule
\multicolumn{7}{c}{$m = n$} \\
\midrule
$n$ & SOCP (eig time) & LTRSR & ADMM & CD-ADMM & Ratio1 & Ratio2 \\
\midrule
1000  & 0.1567 (0.1041) & 0.1052 & \textbf{0.0276} & 0.0731 & 2.14 & 0.26 \\
2000  & 0.7367 (0.3639) & 0.3684 & \textbf{0.1592} & 0.2561 & 2.88 & 0.43 \\
4000  & 3.4745 (1.6937) & 1.9880 & \textbf{1.1545} & 1.6047 & 2.17 & 0.58 \\
6000  & 9.2429 (3.6473) & 4.0068 & \textbf{3.6790} & 4.1341 & 2.24 & 0.92 \\
8000  & 18.6919 (7.1702) & 8.7599 & 8.4552 & \textbf{8.1943} & 2.28 & 0.97 \\
10000 & 31.6802 (13.3527) & 15.1721 & 15.8802 & \textbf{13.6469} & 2.32 & 1.05 \\

\midrule
\multicolumn{7}{c}{$m = 0.5n$} \\
\midrule
$n$ & SOCP (eig time) & LTRSR & ADMM & CD-ADMM & Ratio1  & Ratio2 \\
\midrule
1000  & 0.1320 (0.0921) & \textbf{0.0331} & 0.0421 & 0.0681 & 1.94 & 1.27 \\
2000  & 0.4717 (0.3130) & \textbf{0.1217} & 0.1221 & 0.2500 & 1.89 & 1.00 \\
4000  & 2.7886 (1.6494) & \textbf{0.6481} & 1.0214 & 1.5640 & 1.78 & 1.58 \\
6000  & 7.3335 (3.0921) & \textbf{1.4399} & 3.3971 & 4.0185 & 1.82 & 2.36 \\
8000  & 14.7920 (6.3160) & \textbf{2.5178} & 7.7525 & 7.6182 & 1.94 & 3.08 \\
10000 & 28.8792 (11.1455) & \textbf{3.7470} & 14.6757 & 12.4680 & 2.32 & 3.92 \\

\bottomrule
\end{tabular}
\end{table}

\begin{table}[htbp]
\centering
\caption{Relative error on synthetic data (dense), $\gamma = 0.1$}
\label{tab:relerr_dense_gamma01}
\begin{tabular}{l ccc ccc}
\toprule
\multirow{2}{*}{$\gamma = 0.1$} &
\multicolumn{3}{c}{$(f_{\mathrm{SOCP}}-f_{\mathrm{LTRSR}})/|f_{\mathrm{SOCP}}|$} &
\multicolumn{3}{c}{$(f_{\mathrm{SOCP}}-f_{\mathrm{CD-ADMM}})/|f_{\mathrm{SOCP}}|$} \\
\cmidrule(lr){2-4}\cmidrule(lr){5-7}
& \textbf{AVG} & \textbf{MIN} & \textbf{MAX} & \textbf{AVG} & \textbf{MIN} & \textbf{MAX} \\
\midrule
$m=2n$   & 5.36E-10 & 1.43E-12 & 1.96E-09 & 1.08E-09 & 4.22E-10 & 2.34E-09 \\
$m=n$    & 1.88E-08 & 3.26E-12 & 1.10E-07 & 1.95E-08 & 7.47E-10 & 1.10E-07 \\
$m=0.5n$ & 9.07E-11 & 4.01E-13 & 2.69E-10 & 5.91E-10 & 2.84E-10 & 9.56E-10 \\
\bottomrule
\end{tabular}
\end{table}

\begin{table}[htbp]
\centering
\caption{Relative error on synthetic data (dense), $\gamma=0.01$}
\label{tab:relerr_dense_gamma001}
\begin{tabular}{l ccc ccc}
\toprule
\multirow{2}{*}{$\gamma = 0.01$} &
\multicolumn{3}{c}{$(f_{\mathrm{SOCP}}-f_{\mathrm{LTRSR}})/|f_{\mathrm{SOCP}}|$} &
\multicolumn{3}{c}{$(f_{\mathrm{SOCP}}-f_{\mathrm{CD-ADMM}})/|f_{\mathrm{SOCP}}|$} \\
\cmidrule(lr){2-4}\cmidrule(lr){5-7}
& \textbf{AVG} & \textbf{MIN} & \textbf{MAX} & \textbf{AVG} & \textbf{MIN} & \textbf{MAX} \\
\midrule
$m=2n$   & 2.04E-08 & 3.04E-09 & 4.31E-08 & 2.13E-08 & 4.09E-09 & 4.41E-08 \\
$m=n$    & 2.50E-08 & 5.03E-11 & 5.88E-08 & 2.60E-08 & 1.16E-09 & 5.98E-08 \\
$m=0.5n$ & 5.72E-09 & 8.11E-12 & 2.23E-08 & 6.65E-09 & 9.15E-10 & 2.34E-08 \\
\bottomrule
\end{tabular}
\end{table}

\subsubsection{Sparse Data}
\label{app:sparse}
We process to conduct experiments on a large-scale sparse dataset with varying sparsity levels and sample size  $m \in \{0.5n, n, 2n, 3n\}$. Tables \ref{tab:time_sparse_big_gamma01} and \ref{tab:time_sparse_big_gamma001} report runtimes across a spectrum of sparsity levels ($10^{-2}, 10^{-3}, 10^{-4}$) and penalty parameters ($\gamma=0.1$ and $\gamma=0.01$). A distinct trend is observed: the computational advantage of CD-ADMM becomes increasingly pronounced as sparsity increases. In highly sparse regimes ($10^{-4}$), CD-ADMM effectively leverages the sparse Cholesky factorization, achieving speedups of up to two orders of magnitude over standard ADMM and SOCP. While the iterative solver LTRSR remains competitive in denser settings ($10^{-2}$), CD-ADMM consistently outperforms it in high-dimensional sparse scenarios. Furthermore, the similarity in performance patterns between Table 8 and Table 9 confirms that the proposed method is robust to variations in the hyperparameter $\gamma$, maintaining its efficiency dominance regardless of the penalty weight. Finally, Table \ref{tab:relerr_sparse_big}
similarly demonstrates that, our algorithm achieves precision comparable to state-of-the-art methods, even on sparse datasets.
\begin{table}[htbp]
\centering
\caption{Time (seconds) on synthetic data with different sparsity ($\gamma=0.1$).
Ratio1 $=\mathrm{SOCP}/\mathrm{CD\text{-}ADMM}$ and
Ratio2 $=\mathrm{ADMM}/\mathrm{CD\text{-}ADMM}$.}
\label{tab:time_sparse_big_gamma01}
\resizebox{\textwidth}{!}{%
\begin{tabular}{r cccccc|cccccc|cccccc}
\toprule

\multicolumn{19}{c}{$m=0.5n$} \\
\midrule
& \multicolumn{6}{c}{sparsity $=0.01$}
& \multicolumn{6}{c}{sparsity $=0.001$}
& \multicolumn{6}{c}{sparsity $=0.0001$} \\
\cmidrule(lr){2-7}\cmidrule(lr){8-13}\cmidrule(lr){14-19}
$n$
& SOCP (eig) & LTRSR & ADMM & CD-ADMM & Ratio1 & Ratio2
& SOCP (eig) & LTRSR & ADMM & CD-ADMM & Ratio1 & Ratio2
& SOCP (eig) & LTRSR & ADMM & CD-ADMM & Ratio1 & Ratio2 \\
\midrule
10000 & 21.658 (12.959) & \textbf{0.748} & 13.781 & 3.120 & 6.94 & 4.42
      & 17.247 (11.126) & \textbf{0.592} & 13.046 & 1.281 & 13.46 & 10.18
      & 12.196 (11.623) & 0.653 & 13.390 & \textbf{0.085} & 143 & 157 \\
15000 & 69.211 (41.838) & \textbf{2.625} & 44.617 & 10.262 & 6.75 & 4.35
      & 58.016 (36.962) & \textbf{0.612} & 43.016 & 3.000 & 19.34 & 14.34
      & 34.800 (33.016) & 0.638 & 42.199 & \textbf{0.116} & 301 & 365 \\
20000 & 157.264 (96.395) & \textbf{2.635} & 102.099 & 20.018 & 7.85 & 5.10
      & 148.782 (101.735) & \textbf{0.799} & 98.839 & 7.079 & 21.02 & 13.97
      & 79.709 (72.396) & 0.842 & 100.689 & \textbf{0.310} & 257 & 325 \\
25000 & 303.045 (185.153) & \textbf{3.385} & 194.853 & 35.589 & 8.52 & 5.47
      & 270.733 (179.852) & \textbf{0.831} & 191.138 & 14.148 & 19.15 & 13.51
      & 186.674 (165.716) & 1.090 & 186.911 & \textbf{0.442} & 423 & 423 \\
30000 & 516.413 (315.055) & \textbf{6.479} & 332.583 & 62.835 & 8.22 & 5.29
      & 466.924 (309.113) & \textbf{1.053} & 327.535 & 26.245 & 17.80 & 12.48
      & 342.463 (293.426) & 1.347 & 318.067 & \textbf{0.684} & 501 & 465 \\

\midrule
\multicolumn{19}{c}{$m=n$} \\
\midrule
10000 & 24.235 (14.142) & \textbf{2.445} & 14.087 & 4.821 & 5.03 & 2.92
      & 23.289 (16.135) & \textbf{0.688} & 4.380 & 1.257 & 18.52 & 3.48
      & 14.041 (12.492) & 0.502 & 13.749 & \textbf{0.049} & 285 & 279 \\
15000 & 78.719 (48.351) & \textbf{5.863} & 44.413 & 11.864 & 6.63 & 3.74
      & 69.417 (47.868) & \textbf{1.229} & 42.699 & 3.471 & 20.00 & 12.30
      & 47.203 (41.625) & 0.757 & 43.796 & \textbf{0.179} & 264 & 245 \\
20000 & 176.435 (110.379) & \textbf{11.577} & 102.001 & 23.215 & 7.60 & 4.39
      & 158.147 (108.667) & \textbf{1.997} & 99.463 & 9.054 & 17.47 & 11.00
      & 118.122 (96.957) & 0.976 & 100.073 & \textbf{0.377} & 313 & 266 \\
25000 & 335.268 (206.268) & \textbf{20.036} & 195.198 & 44.206 & 7.58 & 4.42
      & 309.443 (214.607) & \textbf{2.827} & 192.903 & 19.980 & 15.49 & 9.66
      & 236.470 (188.916) & 1.266 & 191.305 & \textbf{0.829} & 285 & 231 \\
30000 & 572.490 (359.249) & \textbf{29.889} & 332.252 & 70.137 & 8.16 & 4.74
      & 517.747 (356.447) & \textbf{3.867} & 330.581 & 36.157 & 14.32 & 9.14
      & 432.633 (333.367) & \textbf{1.694} & 88.595 & 2.006 & 216 & 44 \\

\midrule
\multicolumn{19}{c}{$m=2n$} \\
\midrule
10000 & 26.750 (15.760) & \textbf{2.035} & 14.023 & 5.542 & 4.83 & 2.53
      & 23.480 (16.154) & \textbf{0.391} & 4.211 & 1.304 & 18.01 & 3.23
      & 16.205 (13.393) & 0.701 & 13.726 & \textbf{0.101} & 160 & 136 \\
15000 & 82.755 (50.753) & \textbf{5.143} & 44.379 & 12.983 & 6.37 & 3.42
      & 70.787 (48.468) & \textbf{0.707} & 43.100 & 5.064 & 13.98 & 8.51
      & 55.850 (45.386) & 0.946 & 12.803 & \textbf{0.259} & 215 & 49 \\
20000 & 179.779 (108.653) & \textbf{9.695} & 101.948 & 25.130 & 7.15 & 4.06
      & 159.798 (108.907) & \textbf{1.152} & 100.521 & 13.141 & 12.16 & 7.65
      & 142.832 (109.946) & 1.315 & 27.696 & \textbf{0.811} & 176 & 34 \\
25000 & 333.025 (202.739) & \textbf{14.778} & 189.989 & 47.342 & 7.03 & 4.01
      & 315.903 (217.752) & \textbf{1.712} & 193.358 & 24.579 & 12.85 & 7.87
      & 281.458 (209.865) & \textbf{1.720} & 53.666 & 2.568 & 110 & 21 \\
30000 & 556.265 (337.624) & \textbf{22.828} & 332.847 & 73.648 & 7.55 & 4.52
      & 529.428 (365.183) & \textbf{2.551} & 329.892 & 37.063 & 14.29 & 8.90
      & 481.532 (353.768) & \textbf{2.227} & 90.885 & 5.790 & 83 & 16 \\

\midrule
\multicolumn{19}{c}{$m=3n$} \\
\midrule
10000 & 29.303 (16.817) & \textbf{3.087} & 14.070 & 5.793 & 5.06 & 2.43
      & 21.990 (14.791) & \textbf{0.483} & 4.281 & 1.656 & 13.28 & 2.59
      & 18.850 (14.842) & 0.810 & 13.647 & \textbf{0.109} & 173 & 125 \\
15000 & 84.976 (51.759) & \textbf{7.304} & 44.527 & 13.457 & 6.31 & 3.31
      & 69.951 (46.908) & \textbf{0.915} & 43.519 & 6.469 & 10.81 & 6.72
      & 61.937 (47.430) & 1.204 & 12.798 & \textbf{0.434} & 143 & 29 \\
20000 & 185.337 (114.060) & \textbf{13.939} & 102.332 & 25.340 & 7.31 & 4.04
      & 164.985 (113.201) & \textbf{1.559} & 100.785 & 14.078 & 11.72 & 7.16
      & 150.974 (112.789) & 1.619 & 28.738 & \textbf{1.510} & 100 & 19 \\
25000 & 364.347 (225.367) & \textbf{20.751} & 197.715 & 49.340 & 7.38 & 4.01
      & 323.304 (224.721) & \textbf{2.391} & 190.888 & 23.536 & 13.73 & 8.11
      & 295.949 (218.266) & \textbf{2.118} & 54.505 & 3.837 & 77 & 14 \\
30000 & 606.143 (381.284) & \textbf{30.478} & 326.553 & 73.253 & 8.27 & 4.46
      & 506.082 (339.917) & \textbf{3.532} & 332.294 & 39.415 & 12.84 & 8.43
      & 489.773 (343.090) & \textbf{2.779} & 83.778 & 8.972 & 55 & 9 \\

\bottomrule
\end{tabular}%
}
\end{table}

\begin{table}[htbp]
\centering
\caption{Time (seconds) on synthetic data with different sparsity ($\gamma=0.01$). 
Ratio1 $=\mathrm{SOCP}/\mathrm{CD\text{-}ADMM}$ and Ratio2 $=\mathrm{ADMM}/\mathrm{CD\text{-}ADMM}$.}
\label{tab:time_sparse_big_gamma001}
\resizebox{\textwidth}{!}{%
\begin{tabular}{r cccccc|cccccc|cccccc}
\toprule

\multicolumn{19}{c}{$m=0.5n$} \\
\midrule
& \multicolumn{6}{c}{sparsity $=0.01$} & \multicolumn{6}{c}{sparsity $=0.001$} & \multicolumn{6}{c}{sparsity $=0.0001$} \\
\cmidrule(lr){2-7}\cmidrule(lr){8-13}\cmidrule(lr){14-19}
$n$
& SOCP (eig) & LTRSR & ADMM & CD-ADMM & Ratio1 & Ratio2
& SOCP (eig) & LTRSR & ADMM & CD-ADMM & Ratio1 & Ratio2
& SOCP (eig) & LTRSR & ADMM & CD-ADMM & Ratio1 & Ratio2 \\
\midrule
10000 & 22.237 (11.771) & \textbf{0.570} & 13.518 & 2.981 & 7.46 & 4.53
      & 19.090 (12.014) & \textbf{0.552} & 13.724 & 1.485 & 12.85 & 9.24
      & 12.013 (11.282) & 0.357 & 13.793 & \textbf{0.046} & 261 & 299 \\
15000 & 62.389 (35.280) & \textbf{1.194} & 42.825 & 10.110 & 6.17 & 4.24
      & 57.295 (35.729) & \textbf{0.402} & 42.692 & 3.158 & 18.14 & 13.52
      & 38.071 (35.918) & 0.561 & 43.551 & \textbf{0.097} & 394 & 450 \\
20000 & 140.342 (79.182) & \textbf{1.294} & 101.784 & 20.256 & 6.93 & 5.03
      & 151.050 (102.640) & \textbf{0.658} & 98.581 & 6.506 & 23.22 & 15.15
      & 101.881 (94.509) & 0.757 & 100.477 & \textbf{0.287} & 355 & 350 \\
25000 & 278.835 (162.104) & \textbf{2.071} & 195.302 & 36.283 & 7.68 & 5.38
      & 283.767 (192.517) & \textbf{0.584} & 190.669 & 14.627 & 19.40 & 13.04
      & 205.124 (182.291) & 0.985 & 193.426 & \textbf{0.456} & 450 & 424 \\
30000 & 451.894 (254.626) & \textbf{3.278} & 332.282 & 62.947 & 7.18 & 5.28
      & 473.235 (314.728) & \textbf{0.700} & 327.898 & 24.745 & 19.12 & 13.25
      & 356.263 (306.654) & 1.245 & 329.001 & \textbf{0.780} & 457 & 422 \\

\midrule
\multicolumn{19}{c}{$m=n$} \\
\midrule
10000 & 23.001 (13.225) & \textbf{2.133} & 14.200 & 4.784 & 4.81 & 2.97
      & 23.281 (15.347) & \textbf{0.756} & 4.011 & 1.350 & 17.24 & 2.97
      & 14.316 (12.903) & 0.565 & 13.670 & \textbf{0.057} & 250 & 239 \\
15000 & 71.542 (41.574) & \textbf{5.129} & 44.204 & 12.185 & 5.87 & 3.63
      & 66.713 (44.664) & \textbf{1.244} & 42.969 & 3.577 & 18.65 & 12.01
      & 47.656 (42.067) & 0.911 & 43.584 & \textbf{0.182} & 262 & 240 \\
20000 & 166.537 (96.336) & \textbf{10.129} & 101.758 & 23.135 & 7.20 & 4.40
      & 149.724 (98.697) & \textbf{1.940} & 99.158 & 9.743 & 15.37 & 10.18
      & 116.862 (94.595) & 1.226 & 106.407 & \textbf{0.397} & 295 & 268 \\
25000 & 322.529 (190.283) & \textbf{17.130} & 194.763 & 44.856 & 7.19 & 4.34
      & 284.743 (189.496) & \textbf{2.661} & 189.389 & 18.997 & 14.99 & 9.97
      & 201.070 (151.940) & 1.377 & 190.995 & \textbf{0.962} & 209 & 199 \\
30000 & 546.889 (333.136) & \textbf{25.728} & 331.484 & 64.844 & 8.43 & 5.11
      & 489.148 (321.497) & \textbf{3.640} & 329.250 & 34.344 & 14.24 & 9.59
      & 375.574 (273.864) & \textbf{1.764} & 84.047 & 2.115 & 178 & 40 \\

\midrule
\multicolumn{19}{c}{$m=2n$} \\
\midrule
10000 & 25.146 (13.721) & \textbf{1.413} & 14.004 & 5.390 & 4.66 & 2.60
      & 23.289 (15.424) & \textbf{0.326} & 4.101 & 1.390 & 16.76 & 2.95
      & 15.308 (12.239) & 0.679 & 13.231 & \textbf{0.105} & 145 & 125 \\
15000 & 80.426 (47.778) & \textbf{3.584} & 44.271 & 13.067 & 6.16 & 3.39
      & 73.663 (46.959) & \textbf{0.577} & 43.146 & 5.117 & 14.40 & 8.43
      & 47.757 (36.988) & 0.956 & 11.161 & \textbf{0.283} & 169 & 39 \\
20000 & 180.853 (111.097) & \textbf{6.701} & 101.991 & 25.745 & 7.03 & 3.96
      & 152.717 (102.941) & \textbf{0.923} & 97.745 & 13.436 & 11.36 & 7.27
      & 134.140 (101.145) & 1.371 & 27.368 & \textbf{0.951} & 141 & 29 \\
25000 & 342.930 (208.918) & \textbf{11.008} & 196.258 & 47.649 & 7.20 & 4.12
      & 277.707 (181.536) & \textbf{1.352} & 193.829 & 24.062 & 11.54 & 8.06
      & 280.387 (204.361) & \textbf{1.849} & 51.160 & 2.685 & 104 & 19 \\
30000 & 571.272 (350.460) & \textbf{16.536} & 332.245 & 73.717 & 7.75 & 4.51
      & 515.975 (350.867) & \textbf{1.944} & 330.382 & 36.580 & 14.10 & 9.03
      & 469.252 (342.067) & \textbf{2.127} & 89.099 & 6.233 & 75 & 14 \\

\midrule
\multicolumn{19}{c}{$m=3n$} \\
\midrule
10000 & 28.990 (17.289) & \textbf{2.280} & 14.093 & 5.787 & 5.01 & 2.44
      & 23.354 (16.035) & \textbf{0.403} & 4.640 & 1.769 & 13.20 & 2.62
      & 17.604 (13.047) & 0.805 & 13.644 & \textbf{0.124} & 142 & 110 \\
15000 & 85.736 (52.285) & \textbf{5.376} & 44.484 & 13.552 & 6.33 & 3.28
      & 72.032 (48.975) & \textbf{0.717} & 43.432 & 6.579 & 10.95 & 6.60
      & 60.629 (45.448) & 1.239 & 12.412 & \textbf{0.520} & 109 & 23.86 \\
20000 & 185.056 (113.525) & \textbf{10.160} & 102.261 & 26.542 & 6.97 & 3.85
      & 171.656 (113.202) & \textbf{1.227} & 100.536 & 14.227 & 12.07 & 7.06
      & 153.372 (109.970) & \textbf{1.650} & 27.850 & 1.677 & 91 & 17 \\
25000 & 361.728 (225.719) & \textbf{16.386} & 198.233 & 49.144 & 7.36 & 4.03
      & 327.448 (225.430) & \textbf{1.896} & 193.042 & 22.877 & 14.32 & 8.44
      & 291.419 (211.783) & \textbf{2.085} & 54.146 & 4.450 & 66 & 12 \\
30000 & 611.816 (386.264) & \textbf{24.232} & 338.281 & 75.400 & 8.11 & 4.49
      & 557.815 (386.235) & \textbf{2.666} & 333.868 & 39.535 & 14.11 & 8.45
      & 510.387 (355.942) & \textbf{2.601} & 92.269 & 10.413 & 49 & 9 \\

\bottomrule
\end{tabular}%
}
\end{table}


\begin{table}[htbp]
\centering
\caption{Relative error on synthetic dataset with different sparsity levels.
}
\label{tab:relerr_sparse_big}
\resizebox{\textwidth}{!}{
\begin{tabular}{l ccc ccc l ccc ccc}
\toprule

\multicolumn{14}{c}{sparsity $=0.0001$} \\
\midrule
\multirow{2}{*}{$\gamma = 0.1$} &
\multicolumn{3}{c}{$(f_{\mathrm{SOCP}}-f_{LTRSR})/|f_{\mathrm{SOCP}}|$} &
\multicolumn{3}{c}{$(f_{\mathrm{SOCP}}-f_{CD-ADMM})/|f_{\mathrm{SOCP}}|$} &
\multirow{2}{*}{$\gamma = 0.01$} &
\multicolumn{3}{c}{$(f_{\mathrm{SOCP}}-f_{LTRSR})/|f_{\mathrm{SOCP}}|$} &
\multicolumn{3}{c}{$(f_{\mathrm{SOCP}}-f_{CD-ADMM})/|f_{\mathrm{SOCP}}|$} \\
\cmidrule(lr){2-4}\cmidrule(lr){5-7}\cmidrule(lr){9-11}\cmidrule(lr){12-14}
& AVG & MIN & MAX & AVG & MIN & MAX
& & AVG & MIN & MAX & AVG & MIN & MAX \\
\midrule
$m=3n$   & 3.08E-10 & 2.70E-13 & 1.32E-09 & 3.09E-10 & 3.96E-13 & 1.32E-09
         & $m=3n$   & 3.07E-09 & 2.22E-10 & 6.46E-09 & 3.07E-09 & 2.22E-10 & 6.46E-09 \\
$m=2n$   & 1.05E-09 & 1.79E-11 & 2.93E-09 & 1.05E-09 & 1.81E-11 & 2.93E-09
         & $m=2n$   & 2.36E-09 & 3.52E-11 & 1.10E-08 & 2.37E-09 & 6.75E-11 & 1.10E-08 \\
$m=n$    & 3.88E-09 & 5.26E-13 & 1.92E-08 & 3.90E-09 & 5.86E-13 & 1.92E-08
         & $m=n$    & 6.89E-10 & 8.48E-11 & 1.70E-09 & 3.93E-09 & 8.48E-11 & 1.71E-08 \\
$m=0.5n$ & 2.94E-09 & 2.75E-12 & 1.44E-08 & 4.96E-09 & 7.29E-11 & 1.44E-08
         & $m=0.5n$ & 1.73E-09 & 6.10E-11 & 6.68E-09 & 8.06E-07 & 2.12E-10 & 4.01E-06 \\

\midrule
\multicolumn{14}{c}{sparsity $=0.001$} \\
\midrule
$m=3n$   & 2.76E-08 & 5.21E-10 & 6.51E-08 & 2.78E-08 & 5.28E-10 & 6.53E-08
         & $m=3n$   & 7.07E-08 & 5.18E-10 & 3.11E-07 & 7.12E-08 & 1.29E-09 & 3.12E-07 \\
$m=2n$   & 2.66E-10 & 2.62E-12 & 7.55E-10 & 4.07E-10 & 1.05E-10 & 7.64E-10
         & $m=2n$   & 1.04E-07 & 1.18E-08 & 3.89E-07 & 1.04E-07 & 1.27E-08 & 3.89E-07 \\
$m=n$    & 8.14E-09 & 6.00E-10 & 2.97E-08 & 8.15E-09 & 6.00E-10 & 2.97E-08
         & $m=n$    & 2.87E-08 & 1.53E-08 & 4.32E-08 & 2.88E-08 & 1.53E-08 & 4.32E-08 \\
$m=0.5n$ & 6.78E-10 & 8.28E-13 & 2.42E-09 & 6.80E-10 & 7.32E-13 & 2.43E-09
         & $m=0.5n$ & 1.04E-08 & 7.25E-11 & 2.73E-08 & 1.04E-08 & 7.25E-11 & 2.73E-08 \\

\midrule
\multicolumn{14}{c}{sparsity $=0.01$} \\
\midrule
$m=3n$   & 4.18E-07 & 1.20E-11 & 2.08E-06 & 4.19E-07 & 8.65E-10 & 2.08E-06
         & $m=3n$   & 7.66E-04 & 4.85E-08 & 3.83E-03 & 7.66E-04 & 4.96E-08 & 3.83E-03 \\
$m=2n$   & 1.55E-07 & 8.94E-11 & 7.56E-07 & 1.55E-07 & 1.07E-09 & 7.57E-07
         & $m=2n$   & 1.20E-06 & 1.61E-07 & 3.51E-06 & 1.20E-06 & 1.62E-07 & 3.51E-06 \\
$m=n$    & 2.25E-07 & 2.26E-10 & 1.11E-06 & 2.26E-07 & 1.08E-09 & 1.11E-06
         & $m=n$    & 1.30E-06 & 2.99E-08 & 2.73E-06 & 1.31E-06 & 3.10E-08 & 2.73E-06 \\
$m=0.5n$ & 1.51E-09 & 2.94E-11 & 4.23E-09 & 2.26E-09 & 2.43E-10 & 5.14E-09
         & $m=0.5n$ & 2.24E-05 & 6.07E-14 & 1.11E-04 & 2.24E-05 & 3.38E-11 & 1.11E-04 \\

\bottomrule
\end{tabular}}
\end{table}

\end{document}